\documentclass[a4paper,UKenglish,cleveref]{lipics-v2021}

\usepackage[utf8]{inputenc}
\usepackage{graphicx}
\usepackage{amsfonts}
\usepackage{amsthm}
\usepackage{stackengine,scalerel}
\usepackage{listings}
\usepackage{algorithm}
\usepackage[noend]{algpseudocode}
\usepackage[cppcommentstyle,continuouslinenumbers]{smartalgorithmic}
\usepackage{bm}
\usepackage{makecell}
\usepackage{pifont}

\definecolor{mediumgreen}{RGB}{0,150,0}
\definecolor{sky}{RGB}{50,150,255}

\newcommand{\myparagraph}[1]{\noindent{\textbf{#1.}}}

\NewDeclarationType{DistributedObject}{\textsc}
\NewDeclarationType{Function}{\textsf}
\NewDeclarationType{Event}{\textsf}
\NewDeclarationType{Operation}{\textbf}
\NewDeclarationType{Variable}{\mathit}
\NewDeclarationType{Property}{\textrm}
\NewDeclarationType{Type}{\textrm}
\NewDeclarationType{MessageType}{\textbf}
\NewDeclarationType{Constant}{\textbf}
\NewDeclarationType{Term}{}

\mathchardef\mhyphen="2D
\newcommand{\cL}{\mathcal{L}}
\newcommand{\cP}{\mathcal{P}}

\newcommand{\Message}[1]{\ensuremath{\langle}#1\ensuremath{\rangle}}

\algblockdefx[Multiline]{Multiline}{EndMultiline}{}{}
\algtext*{EndMultiline}  

\algblockdefx[Function]{Function}{EndFunction}[2]{\textbf{function} #1(#2)}{}
\algtext*{EndFunction}  

\algblockdefx[ForEach]{ForEach}{EndForEach}[1]{\textbf{for each} #1 \textbf{do}}{}
\algtext*{EndForEach}  

\algblockdefx[Repeat]{Repeat}{EndRepeat}[1]{\textbf{repeat} #1}{}
\algtext*{EndRepeat}

\algblockdefx[Procedure]{Procedure}{EndProcedure}[2]{\textbf{procedure} #1(#2)}{}
\algtext*{EndProcedure}  

\algblockdefx[Operation]{Operation}{EndOperation}[2]{\textbf{operation} #1(#2)}{}
\algtext*{EndOperation}  

\algblockdefx[OverrideOperation]{OverrideOperation}{EndOperation}[2]{\textbf{override operation} #1(#2)}{}
\algtext*{EndOperation}  

\algblockdefx[Upon]{Upon}{EndUpon}[1]{\textbf{upon} #1}{}
\algtext*{EndUpon}  

\algblockdefx[Behavior]{Behavior}{EndBehavior}[1]{\textbf{behavior} of #1}{}
\algtext*{EndBehavior}  

\algblockdefx[UponEvent]{UponEvent}{EndUponEvent}[1]{\textbf{upon event} #1}{}
\algtext*{EndUponEvent}  

\algblockdefx[UponReceiving]{UponReceiving}{EndUponReceiving}[2]{\textbf{upon} \Receive~\Message{#1} \textbf{from} #2}{}
\algtext*{EndUponReceiving}  

\algblockdefx[UponConsistentDeliver]{UponConsistentDeliver}{EndUponConsistentDeliver}[2]{\textbf{upon consistently delivering} $\langle$#1$\rangle$ \textbf{from} #2}{}
\algtext*{EndUponConsistentDeliver}  

\algblockdefx[ProcessState]{ProcessState}{EndProcessState}{\textbf{State:}}{}
\algtext*{EndProcessState}  

\algblockdefx[DistributedObjects]{DistributedObjects}{EndDistributedObjects}{\textbf{Distributed objects:}}{}
\algtext*{EndDistributedObjects}  

\algblockdefx[BehaviorLA]{BehaviorLA}{EndBehaviorLA}{\textbf{Behavior upon receiving a message $m$ of instance $LA_k$:}}{}
\algtext*{EndBehaviorLA}  

\algblockdefx[FunctionDeclaration]{FunctionDeclaration}{EndFunctionDeclaration}{\textbf{Functions:}}{}
\algtext*{EndFunctionDeclaration}

\algblockdefx[StaticLocalVariables]{StaticLocalVariables}{EndStaticLocalVariables}{\textbf{Static Local Variables:}}{}
\algtext*{EndStaticLocalVariables}  

\algnewcommand{\Returns}[1]{\textbf{returns} {#1}}

\algnewcommand{\WaitFor}{\textbf{wait for} }
\algnewcommand{\WaitUntil}{\textbf{wait until} }
\algnewcommand{\Send}[2]{\textbf{send} $\langle$#1$\rangle$ \textbf{to} #2}
\algnewcommand{\IfThenElse}[3]{
    \algorithmicif\ #1\ \algorithmicthen\ #2\ \algorithmicelse\ #3}
\algnewcommand{\IfThen}[2]{
    \algorithmicif\ #1\ \algorithmicthen\ #2}

\algnewcommand{\TriggerEvent}[1]{\textbf{trigger event }#1}

\algnewcommand{\Parameters}[1]{\State \textbf{Instance parameters: }#1}

\DeclareGlobalVariable{common}
\DeclareGlobalVariable{Last}
\DeclareGlobalVariable{Prop}
\DeclareGlobalVariable{MPool}
\DeclareGlobalVariable[sskip]{skip} 
\DeclareGlobalVariable{reply}
\DeclareGlobalVariable{Pending}
\DeclareGlobalVariable{Validated}
\DeclareGlobalVariable{Proposing}
\DeclareGlobalVariable{Learned}
\DeclareGlobalVariable{ApCall}
\DeclareGlobalFunction{keys}
\DeclareGlobalFunction{add}

\DeclareGlobalMessageType{POLL}
\DeclareGlobalMessageType{SKIP}
\newcommand{\PREPARE}{\textbf{PROPOSE}}
\DeclareGlobalMessageType{PROPOSE}
\DeclareGlobalMessageType{ACCEPT}
\DeclareGlobalMessageType{REQUEST}

\DeclareGlobalDistributedObject{GLA}
\DeclareGlobalDistributedObject{LA}
\DeclareGlobalDistributedObject{ASO}
\DeclareGlobalDistributedObject{AS}
\DeclareGlobalDistributedObject{SCD}
\DeclareGlobalDistributedObject{Broadcast}

\DeclareGlobalFunction{Propose}
\DeclareGlobalFunction{SendRequest}
\DeclareGlobalFunction{Update}
\DeclareGlobalFunction{Snapshot}

\DeclareGlobalProperty[OValidity]{1-Validity}
\DeclareGlobalProperty[OStability]{1-Stability}
\DeclareGlobalProperty[OConsistency]{1-Consistency}
\DeclareGlobalProperty[OLiveness]{1-Liveness}

\DeclareGlobalProperty[GValidity]{Validity}
\DeclareGlobalProperty[GStability]{Stability}
\DeclareGlobalProperty[GConsistency]{Consistency}
\DeclareGlobalProperty[GLiveness]{Liveness}

\DeclareGlobalProperty[StrongValidity]{Strong Validity}
\DeclareGlobalProperty[QuorumAwareness]{Quorum Awareness}

\DeclareGlobalConstant{false}
\DeclareGlobalConstant{true}

\DeclareGlobalEvent{Receive}
\DeclareGlobalEvent{Learn}
\DeclareGlobalEvent{Startup}

\newcommand{\UseBoldForGoodMathEntries}{1}
\newcommand{\UseBoldForGoodTextEntries}{1}

\newcommand{\cI}{\mathcal I}
\newcommand{\cO}{\mathcal O}
\newcommand{\cQ}{\mathcal Q}

\newcommand{\cT}{\mathcal T}

\newcommand{\Start}{\textit{start}}
\newcommand{\End}{\textit{end}}
\newcommand{\Duration}{\textit{duration}}

\newcommand{\OpCall}{\textit{aReq}}
\newcommand{\OpResp}{\textit{aRep}}
\newcommand{\Op}{\textit{op}}
\newcommand{\Data}{\textit{data}}
\newcommand{\NTR}{\textbf{NTR}}
\newcommand{\LCC}{\textbf{LCC}}
\newcommand{\IRA}{\textbf{IRA}}
\newcommand{\CR}{\textbf{CR}}
\newcommand{\Round}{\textbf{Round}}

\DeclareGlobalFunction{send}
\DeclareGlobalFunction{receive}
\DeclareGlobalFunction{update}
\DeclareGlobalFunction{snapshot}
\DeclareGlobalFunction{interval}
\DeclareGlobalEvent{OK}
\DeclareGlobalConstant{SW}
\DeclareGlobalConstant{MW}
\DeclareGlobalConstant{SR}
\DeclareGlobalConstant{MR}
\DeclareGlobalVariable{acceptedValue}
\DeclareGlobalVariable{Relaying}
\DeclareGlobalVariable{ToAdopt}

\newcommand{\remove}[1]{}

\definecolor{mediumgreen}{RGB}{0,150,0}
\definecolor{darkgreen}{RGB}{0,100,0}
\definecolor{sky}{RGB}{50,150,255}

\newcommand{\GoodEntry}[1]{\textcolor{mediumgreen}{\ifthenelse{\UseBoldForGoodTextEntries=1}{\textbf{#1}}{#1}}}
\newcommand{\GoodMathEntry}[1]{\textcolor{mediumgreen}{\ensuremath{\ifthenelse{\UseBoldForGoodMathEntries=1}{\bm{#1}}{#1}}}}
\newcommand{\AlmostGoodEntry}[1]{\textcolor{orange}{\ifthenelse{\UseBoldForGoodTextEntries=1}{\textbf{#1}}{#1}}}
\newcommand{\AlmostGoodMathEntry}[1]{\textcolor{orange}{\ensuremath{\ifthenelse{\UseBoldForGoodMathEntries=1}{\bm{#1}}{#1}}}}

\def\REVIEW{1}  

\if \REVIEW 1
    \usepackage[enable]{xnotes}
\else
    \usepackage[apply,nonotes]{xnotes}
\fi

\AddXNotesUser{jp}{João}{purple}
\AddXNotesUser{lf}{Luciano}{blue}
\AddXNotesUser{pk}{Petr}{teal}

\title{Asynchronous Latency and Fast Atomic Snapshot}
\author{João Paulo Bezerra}{LTCI, Télécom Paris, Institut Polytechnique de Paris, France}{joaopba01@gmail.com}{}{}
\author{Luciano Freitas}{Nomadic Labs, Paris, France}{lucianofdz@gmail.com}{}{}
\author{Petr Kuznetsov}{LTCI, Télécom Paris, Institut Polytechnique de Paris, France}{petr.kuznetsov@telecom-paris.fr}{}{}
\author{Matthieu Rambaud}{LTCI, Télécom Paris, Institut Polytechnique de Paris, France}{matthieu.rambaud@telecom-paris.fr}{}{}

\Copyright{J. P. Bezerra, L. Freitas, P. Kuznetsov, and M. Rambaud}

\authorrunning{J. P. Bezerra, L. Freitas, P. Kuznetsov, and M. Rambaud}


\nolinenumbers

\begin{document}

\maketitle

\begin{abstract}
This paper introduces a novel, \emph{fast} atomic-snapshot protocol for asynchronous message-passing systems. 
In the process of defining what ``fast'' means exactly, we spot a few interesting issues that arise when conventional time metrics are applied to long-lived asynchronous algorithms.   
We reveal some gaps in latency claims made in earlier work on snapshot algorithms, which hamper their comparative time-complexity analysis.   
We then come up with a new unifying time-complexity metric that captures the latency of an operation in an asynchronous, long-lived implementation. 
This allows us to formally grasp latency improvements of our atomic-snapshot algorithm with respect to the state-of-the-art protocols: optimal latency in fault-free runs without contention, short constant latency in fault-free runs with contention, the worst-case latency proportional to the number of active concurrent failures, and constant, \emph{amortized} latency. 
\end{abstract}

\section{Introduction}
The \emph{distributed snapshot} abstraction~\cite{chandyL85as,mattern93as} allows us to determine a consistent view of the global system state.
Originally proposed in the asynchronous fault-free message-passing context, it was later cast to shared-memory models~\cite{atomic-snapshot} as a vector of shared variables, exporting an \emph{update} operation that writes to one of them and a \emph{snapshot} operation that returns the current vector state.
%
Atomic snapshot can be implemented from conventional read-write registers in a \emph{wait-free} manner, i.e., tolerating unpredictable delays or failures of any number of processes.    
By applying the reduction from shared memory to message-passing~\cite{ABDpaper}, one can get an asynchronous distributed atomic-snapshot implementation that tolerates up to a minority of faulty processes.  
The \emph{atomic-snapshot object} (ASO) is, in a strong sense, equivalent to \emph{lattice agreement} (LA)~\cite{la95,faleiro2012podc}%
\footnote{Lattice agreement can be seen as a weak version of consensus, where decided values form totally ordered \emph{joins} of proposed values in a \emph{join semi-lattice}.}: one can implement the other with no time overhead.
A long line of results improve time and space complexities of ASO and LA algorithms in shared-memory~\cite{AspnesC13,AspnesACE15,AttiyaEF11} and message-passing~\cite{faleiro2012podc,imbs2018set,ASOconstant,delporte2018implementing,garg2020amortized} models.   

In this paper, we focus on the \emph{latency} of operations in message-passing ASO implementations. 
We propose an LA (and, thus, ASO) algorithm that is \emph{faster} than (or matches) state-of-the-art solutions in all execution scenarios: with or without failures and with or without contention.  
The comparative analysis of our algorithm with respect to the existing work appeared to be challenging: as we show below, earlier work considered diverging metrics and execution scenarios, and sometimes used over-simplified reasoning.
We observed that conventional metrics~\cite{canetti1993fast,ABDpaper,abraham-metric} are not always suitable for \emph{long-lived} asynchronous algorithms. 
Besides, prior latency analyses of ASO and LA algorithms~\cite{faleiro2012podc,imbs2018set,ASOconstant,delporte2018implementing,garg2020amortized} used different ways to measure time, which complicated the comparison.  
We therefore propose a unifying time-complexity analysis of prior asynchronous ASO and LA algorithms with respect to a new metric, which we take as a contribution on its own.

\begin{table*}[htp]
     \begin{center}
  \begin{tabular}{||c|c|c|c|c||}
  \hline
   & \makecell[c]{Fault-free \\ w/o contention} & \makecell[c]{Fault-free \\ w/ contention} & Worst-case & \makecell[c]{Amortized \\ constant}\\
  \hline
 
  \hline
 Faleiro et al.~\cite{faleiro2012podc} & $2$ & $16$ & $O(k)$ & yes\\ \hline
 Imbs et al.~\cite{imbs2018set} & $2$ & $O(n)$ & $O(n)$ & no\\ \hline
 Garg et al.~\cite{garg2020amortized} & $\ge 6$ & $\ge 8$ & $O(k)$ & yes \\ \hline
 Garg et al.~\cite{garg2020amortized} + Zheng et al.~\cite{zheng2019opodis} & $O(\log n)$ & $O(\log n)$ & $O(\log n)$ & no\\ \hline
 Delporte et al.~\cite{delporte2018implementing} & $2$ & $O(n)$ & $O(n)$ & no \\ \hline
 This paper & $2$ & $8$ & $O(k)$ & yes\\ \hline
 \end{tabular}
 \end{center}
     \caption{Comparative time complexity of atomic-snapshot algorithms in asynchronous message-passing models. The table shows results for Single-Writer Multi-Reader ($\SW\MR$) implementations.}
     \label{tab:compASO}
\end{table*}

%
Lamport~\cite{lamport2006lower} proposed to measure time in asynchronous systems as the length of the longest chain of \emph{causally related} messages, the metric used to to determine the best-case latency of consensus~\cite{lamport2006lower} and Crusader Agreement~\cite{abraham2023round}.  
However, as we show in this paper, the metric may produce counter-intuitive results for protocols involving all-to-all communication. For instance, in the failure-free case, the $n$-process reliable-broadcast~\cite{cachin2011introduction} exhibits a causal chain of $n$ hops, even though, intuitively one expects it to terminate in one.     

Building upon the classical approach by Canetti and Rabin~\cite{canetti1993fast}, Abraham et al.~\cite{abraham-metric} recently proposed an elegant metric to grasp the good-case latency of broadcast protocols. 
We observe, however, that the metric does not really apply to executions of \emph{long-lived} abstractions, which may contain \emph{holes} -- periods of inactivity when no protocol messages are in transit.
Moreover, we get diverging results when applying~\cite{abraham-metric} and~\cite{canetti1993fast} to \emph{operation latency}, i.e., the time between invocation and response events of a given operation.

We therefore extend the round-based approach to long-lived abstractions (such as ASO and LA) and establish a framework to measure the time between arbitrary events, subsequently showing that the results align with those from earlier classical metrics \cite{ABDpaper,canetti1993fast}.

To summarize, our main contribution is a novel LA (and, thus, ASO) protocol that is generally \emph{faster} than prior solutions, i.e., it exhibits shorter latency of its operations in various scenarios.
In our complexity analysis, we compared our protocol  to the original long-lived LA algorithm by Faleiro et al.~\cite{faleiro2012podc}\footnote{We consider the ASO protocol built atop the lattice agreement protocol proposed in~\cite{faleiro2012podc}.}, the first direct message-passing ASO implementation by Delporte et al.~\cite{delporte2018implementing}, the ASO algorithm based on the \emph{set-constraint broadcast} by Imbs et al.~\cite{imbs2018set}, and the ASO algorithms by Garg et al. based on generic construction of ASO from \emph{one-shot} LA with constant latency in fault-free runs~\cite{garg2020amortized} or $\log{n}$ worst-case latency by Zheng et al.~\cite{zheng2019opodis} (where $n$ is the number of processes).

As shown in Table~\ref{tab:compASO},   
in a fault-free run, the latency of an operation of our protocol is the optimal two rounds if there is no contention and eight rounds in the presence of contention 
(four rounds if we ignore the “buffering” period when a value is submitted but not yet proposed),
regardless of the number of contending operations.
Moreover, the worst-case latency of our algorithm is proportional to the number of \emph{active failures} $k$, i.e., the number of faulty processes whose messages are received within the operation's interval, therefore the \emph{amortized} latency (averaged over a large number of operations in a long-lived execution) converges to the fault-free constant.  

Our protocol can be seen as a novel combination of techniques employed separately in prior work.
These include the use of \emph{generalized} (long-lived) lattice agreement as a basis for ASO~\cite{kuznetsov2019opodis}, the \emph{helping} mechanism where all the learned lattice values are shared~\cite{kuznetsov2019opodis}, relaying of messages to all replicas instead of quorum-based rounds~\cite{imbs2018set,garg2020amortized,dag21,narwal22}, and buffering proposed values until previous proposals get committed~\cite{faleiro2012podc}.
Similar to earlier proposals~\cite{faleiro2012podc}, our algorithm involves $O(n^2)$ (all-to-all) communication, which is compensated by its constant (amortized) latency.   
An interesting open question is whether one can reduce the communication cost in good runs, while maintaining constant amortized latency. 

The paper is organized as follows. In Section~\ref{sec:systemModel}, we present our model assumptions, and in Section~\ref{sec:definitions}, we state the problem of atomic snapshot and relate it to generalized lattice agreement. 
In Section~\ref{sec:ourProtocol}, we present our protocol and analyze its correctness. 
In Section~\ref{sec:measuringASO}, we discuss several gaps in the complexity analyses of earlier work. 
In Section~\ref{sec:timeMeasure}, we present a comparative analysis of time metrics. 
Certain proofs and a detailed discussion of time complexity of earlier protocols are delegated to the appendix.

\section{System Model}
\label{sec:systemModel}
\myparagraph{Processes and Channels}
We consider a system of $n$ \emph{processes} (or \emph{nodes}).
Processes communicate by exchanging messages $m = (s,r,\Data)$ with a sender $s$, a receiver $r$, and a message content $\Data$.

A process is an automaton modeled as a tuple $(\cI,\cO,\cQ,q_0,\pi)$, where $\cI$ is a set of inputs (messages and application calls) it can receive, $\cO$ is a set of outputs (messages and application responses), $\cQ$ is a (potentially infinite) set of possible internal states, $q_0 \in \cQ$ is an initial state and $\pi: 2^{\cI}\times\cQ \rightarrow 2^{\cO}\times\cQ$ is a transition function mapping a set of inputs and a state to a set of outputs and a new state.
Each process $i$ is assigned an algorithm $A_i$ which defines $(\cI,\cO,\cQ,q_0,\pi)$, a \emph{distributed algorithm} is an array $[A_1,...,A_n]$.

\myparagraph{Events and Configurations}
Application calls and responses are tuples $(i,\OpCall)$ and $(i,\OpResp)$ with a process identifier, a request, and a reply respectively.

An \emph{event} $e$ is a tuple $(R,P,S)$ where $R$ is a set of received messages and/or application calls, $P$ is the set of nodes producing the event and $S$ is a set of messages sent and/or application responses.
We denote $\receive(e)$ as the set of messages received in the event, conversely, $\send(e)$ is the set of messages sent.
A \emph{message hop} is a pair $(e, e')$ in which $e'$ receives at least one message that was sent in $e$.

Messages in transit are stored in the \emph{message buffer}.\footnote{We assume that every message in the message buffer is unique.}
A \emph{configuration} $C$ is an $(n+1)$-array $[M,s_1,...,s_n]$ with the buffer's state $M=C[0]$ and the local state $s_i=C[i]$ of each node $i$ ($i=1,\ldots,n$).
Let $C_0$ denote the \emph{initial configuration} in which every $s_i$ is an initial state and the buffer $M$ is empty.

\myparagraph{Executions}
An \emph{execution} (or \emph{run}) is an alternating sequence $C_0e_1C_1e_2...$ of configurations and events, where for each $j>0$ and $i=1,\ldots,n$:

\begin{enumerate}
    \item $\receive(e_j) \subseteq C_{j-1}[0]$;
    \item $e_j.S$ consists of messages and application outputs that the nodes in $e_j.P$ produce, given their algorithms, their states in $C_{j-1}$ and their inputs in $e_j.R$; the nodes in $e_j.P$ carry their states from $C_{j-1}$ to $C_j$, accordingly;     
    \item for the nodes $i\notin e_j.P$, $C_{j-1}[i]=C_j[i]$.  
\end{enumerate}

Each triple $C_{j-1}e_jC_j$ is called a \emph{step}.
In this paper, we consider algorithms defined by deterministic automata, and we assume a default initial state. 
Thus, we sometimes skip configurations and simply write $e_1 e_2 \ldots$. 

In an \emph{infinite} execution, a process is \emph{correct} if it takes part in infinitely many steps, and 
\emph{faulty} otherwise.
We only consider infinite executions in which $f < n/2$, where $f$ is the number of faulty processes and $n$ is the total number of processes.
Moreover, in an infinite execution, messages exchanged among correct processes are eventually received, i.e., if there is an event $e$ from a correct process sending a message $m$ to another correct process, then there is $e'$ succeeding $e$ such that $m \in \receive(e')$.

We also assume that the communication channels neither alter nor create messages.  
Finally, we assume that the channels are \emph{FIFO}: messages from a given source to a given destination arrive in the order they were sent.
A FIFO channel can be implemented by attaching sequence numbers to messages, without extra communication or time overhead.

\section{Lattice Agreement and Atomic Snapshot}
\label{sec:definitions}
 
\subsection{Lattice Agreement}
A \emph{join semi-lattice} is defined as a tuple $(\cL, \sqsubseteq)$, where  $\sqsubseteq$ is a partial order on a set $\cL$, such that 
for any pair of values $u$ and $v$ in $\cL$, there exists a unique \emph{least upper bound} $u \sqcup v \in \cL$ ($\sqcup$ is called the \emph{join} operator).
Also, $u$ and $v$ in $\cL$ are said to be \emph{comparable} if $u \sqsubseteq v \vee v \sqsubseteq u$.

The (generalized) \emph{Lattice Agreement} abstraction {\LA}~\cite{faleiro2012podc} defined over $(\cL, \sqsubseteq)$ can be accessed by every node with operation $\Propose(v)$, $v\in\cL$ (we say that the node \emph{proposes $v$}) which triggers the reply event $\Learn(w)$ (we say that the node \emph{learns $w$}).
Each node may invoke $\Propose$ any number of times but does so \emph{sequentially}, that is, it initiates a new operation only after the previous one has returned.\footnote{Following~\cite{kuznetsov2019opodis}, without loss of generality, we slightly modified the conventional LA interface~\cite{la95,faleiro2012podc} by introducing the explicit $\Propose$ operation that combine proposing and learning the values, the properties of the abstraction are adjusted accordingly.}
The abstraction must satisfy:

\begin{definition}[Lattice Agreement (LA)]
\label{def:LA}
\leavevmode
\begin{itemize}
    \item{\bf \GValidity.} Any value learned by a node is the join of some set of proposed values that includes its last proposal.
    \item{\bf \GStability.} The values learned by any node increase monotonically, with respect to $\sqsubseteq$.
    \item{\bf \GConsistency.} All values learned are comparable, with respect to $\sqsubseteq$.
    \item{\bf \GLiveness.} If a correct node proposes $v$, it eventually learns a value $w$.
\end{itemize}
\end{definition}

\subsection{Atomic Snapshot Object (ASO)}

An atomic snapshot object (ASO) stores a vector of values $R = [r_1,...,r_m]$ and exports two operations: $\update(i,v)$ and $\snapshot()$.
The $\update(i,v)$ operation writes the value $v$ in $R[i]$ and returns $\OK$, and $\snapshot()$ returns the entire vector $R$.
An ASO implementation guarantees that every operation invoked by a correct process eventually completes.
It also ensures that each of its operations appears to take effect in a single instance of time within its interval, i.e., it is \emph{linearizable}~\cite{HW90-lin}.

\myparagraph{Linearizable executions}
The \emph{history} of an execution $E$ is the subsequence of $E$ consisting of invocations and responses of ASO operations (update and snapshot). 
A history is \emph{sequential} if each of its invocations is followed by a matching response.
An execution is linearizable if, to each of its operation (update or snapshot, except, possibly, for incomplete ones), we can assign an indivisible point within its interval (called a \emph{linearization point}), so that the operations put in the order of its linearizaton points constitute a \emph{legal} sequential history of ASO (called a \emph{linearization}), i.e., every snapshot operation returns a vector where every position contains the last value written to it (using an update operation), or the initial value if there are no such prior updates.
Equivalently, a linearizable execution $E$ with history $H$ should have a linearization $S$, a legal sequential history that (1)~no node can locally distinguish a completion of $H$ and $S$ and (2)~$S$ respect the \emph{real-time order} of $H$, i.e., if operation $op$ completes before operation $op'$ in $H$, then $op'$ cannot precede $op$ in $S$.

We say that an ASO is \emph{single writer} $\SW$ (resp. \emph{multi writer} $\MW$) if for each of its registers $R[i]$, only a single process can call $\update(i,v)$ (resp. every process can call $\update(i,v)$).
In this paper, we focus mostly on $\SW\MR$ atomic snapshot objects. In Table~\ref{tab:compASO} we give results only for $\SW\MR$. A $\MW\MR$ $\ASO$ can be devised from $\SW\MR$ by adding an additional ``read'' phase when updating values (see Section~\ref{subsec:LA-ASO} for more details).

Next, we show that ASO can be implemented on top of LA with no additional overhead.

\subsection{From LA to ASO}
\label{subsec:LA-ASO}

To implement a $\SW\MR$ $\ASO$ on top of $\LA$, we consider a partially ordered set $\cL^*$ of $(m+n)$-vectors (recall that $m$ is the size of the ASO vector and $n$ is the number of nodes), defined as follows. 

A vector position $\ell\in 1,\ldots, m$ is  defined as a tuple $(w,v)\in R_{\ell}$, where $v$ is an element of a value set $V$ equipped with a total order $\leq^V$, and $w \in \mathbb N$ is the number of write operations on position $\ell$.
A total order on $R_{\ell}$ is defined in the natural way: for any two tuples $(w_1,v_1) \leq^{R_{\ell}} (w_2,v_2) \equiv (w_1 < w_2) \vee (w_1=w_2 \wedge v_1 \leq^V v_2)$.
For each process $i=1,\ldots,n$, the vector position $m+i$  stores the number of snapshot operations executed by $i$.

The lattice $\cL^*$ of $(m+n)$-position vectors is then the composition $R_1\times\ldots \times R_{m} \times \mathbb N^n$.  
The partial order $\sqsubseteq^*$ 
on $\cL^*$ is then naturally defined as the compositions of $<^{R_{1}}\times\ldots\times <^{R_{m}}\times \leq^n$.
The composed join operator $\sqcup^*$ is the composition of $\max$ operators, one for each position in the $(m+n)$-position vectors.
The construction implies a join semi-lattice~\cite{kuznetsov2019opodis}.  

In Algorithm~\ref{alg:GLAtoAS}, we show how to implement an $\SW\MR$ atomic snapshot on top of $\LA$ defined over the semi-lattice $(\cL^*,\sqsubseteq^*,\sqcup^*)$.
For simplicity, we assume that $m = n$, i.e., the size of the array is the total number of nodes, and that each node $i$ has a dedicated register $i$ where it can write. 
Elements of $\cL^*$ are then $2n$-vectors. 

When a node $i$ calls $\update($i,v$)$, it increments its local \emph{writing} sequence number $w$ and proposes a $2n$-vector with $(w,v)$ in position $i$ and initial values in all other positions to the $\LA$ object. 
The vector learned from this proposal is ignored.
When the node $i$ calls $\snapshot()$, it increments its local \emph{reading} sequence number $r$ proposes a $2n$-vector with $r$ in position $n+i$ and initial values in all other positions to the $\LA$ object.
The values in the first $n$ positions of the returned vector is then returned as the snapshot outcome.   

\begin{algorithm}
\begin{smartalgorithmic}[1]

\DistributedObjects
    \State $\LA$ instance on $(\cL^*,\sqsubseteq^*,\sqcup^*)$
\EndDistributedObjects

\algspace[0.5]

\Upon{startup}
    \State $w \gets 0$
    \State $r \gets 0$
\EndUpon

\algspace[0.3]

\Operation{\update}{$i,v$}
    \State $w \gets w + 1$
    \State $V \gets$ $2n$-vector with $(w,v)$ in position $i$ and initial values in all other positions
    \State \LA.\Propose$(V)$
\EndOperation

\algspace[0.3]

\Operation{\snapshot}{}
    \State $r \gets r + 1$
    \State $V \gets$ $2n$-vector with $r$ in position $n+i$ and initial values in all other positions
    \State \Return $(\LA.\Propose(V))[1..n]$
\EndOperation

\end{smartalgorithmic}
\caption{$\LA \to \SW\MR$ $\ASO$ conversion.}
\label{alg:GLAtoAS}
\end{algorithm}

Algorithm~\ref{alg:GLAtoAS} can be extended to implement a $\MW\MR$ $\ASO$: to update a position $j$ in the array, a node first takes a snapshot to get the current state, gets up-to-date sequence number in position $j$ and proposes its value with a higher sequence number.
With this modification, the $\update$ operation takes two $\LA$ operations instead of one.
We refer the reader to~\cite{kuznetsov2019opodis} for further details.

\begin{theorem}
Algorithm~\ref{alg:GLAtoAS} implements {\ASO}. 
\end{theorem}
\begin{proof}
We show that every execution of Algorithm~\ref{alg:GLAtoAS} is linearizable.          

Consider an execution of Algorithm~\ref{alg:GLAtoAS}, let $H$ be its history. 
Every operation (snapshot or update) is associated with a unique sequence number and performs a $\Propose$ operation on the $\LA$ object.
If there is an $\LA.\Propose$ operation that returns $(w,v)$ in position $i$, by {\GValidity} of $\LA$, there is an operation $\update(i,v)$ executed by node $i$ with sequence number $w$ that started before the $\LA.\Propose$ completed and invoked a $\LA.$. 
In this case, we say that the $\update$ operation is \emph{successful}.
Notice that by {\GValidity} of $\LA$, the update must have invoked $\LA.\Propose$ with a vector containing $(w,v)$ in position $i$.   

Now we order complete snapshot operations and complete successful $\update$ operations in the order of the values returned by their $\LA.\Propose$ operations (by {\GConsistency} of $\LA$, these values are totally ordered. 
As each of these $\LA.\Propose$ returns a value containing its unique sequence number ({\GStability} of $\LA$) , this order respects the real-time order of $H$.  
A successful $\update$ operation performed by node $i$ with $(w,v)$ in position $i$ that has no complete $\LA.\Propose$ is placed right before the first $\snapshot$ whose $\LA.\Propose$ returns this value.   
By construction, the resulting sequential history is legal and locally indistinguishable from a completion of $H$.  

Finally, ${\GLiveness}$ implies that every operation invoked by a correct process eventually completes.
\end{proof}

\section {LA Protocol}
\label{sec:ourProtocol}

In Algorithm~\ref{alg:longLA}, we describe our protocol for solving $\LA$.
To guarantee amortized constant complexity, the protocol relies on two basic mechanisms, employed separately in earlier work~\cite{faleiro2012podc,kuznetsov2019opodis}. 
First, when a node receives a request (e.g., a value from the application), it first adds the request to a buffer ($\MPool$) and then relays it before starting a proposal.
This ensures that ``idle'' nodes also help in committing the request.
Second, the node relays every learned value so that nodes that are ``stuck'' can adopt values from other nodes.

\subsection{Overview}
The protocol is based on helping: every node tries to commit every proposed value it is aware of. 
As long as the node has \emph{active} proposals that are not yet committed, it buffers newly arriving proposals in the local variable $\MPool$. 
Intuitively, in the worst case, an $\LA.\Propose$ operation has to wait until \emph{one} of the concurrently invoked $\LA.\Propose$ operations complete. 
Once this happens, the currently buffered value is put in the local dictionary $\Pending$ and shared with the other nodes (lines~\ref{line:pending} and~\ref{line:propose1}) via a $\PREPARE$ message.
In turn, the other nodes relay the message to each other (line~\ref{line:prepM}).
The dictionary maps a value to the number of times it is "supported" by the nodes (using $\PREPARE$ messages).
Once a value $v$ in the dictionary assembles a quorum of $n-f$ of $\langle \PREPARE, v \rangle$ messages, i.e., $\Pending[v]\geq n-f$ (line~\ref{line:valValue}), the value is added to the $\Validated$ variable.
Once every value currently stored in $\Pending$ is in $\Validated$ (line~\ref{line:learnCond1}), the operation completes with $\Validated$ as the learned value.
As the final element of the helping mechanism, each process broadcasts every value it learns (lines~\ref{line:sendAcc1} and~\ref{line:endLearnCond2}), ensuring that processes that might otherwise remain ``stuck'' can complete their current proposal.

In summary, the algorithm relies on four main ideas: 1) buffering incoming requests when already proposing, 2) sharing every received proposal so all processes are quickly aware of active ones, 3) initiating a new proposal only after all currently seen proposals have been validated, and 4) broadcasting learned values to help other processes make progress.

\myparagraph{Message complexity} The protocol is comprised of three all-to-all communication phases: processes send and relay requests at lines~\ref{line:sendReq1} and~\ref{line:sendReq2}, proposals at lines~\ref{line:sendProp1} and~\ref{line:prepM}, and accepted values at lines~\ref{line:sendAcc1} and~\ref{line:endLearnCond2}.
The total number of messages is therefore $O(n^2)$.
However, a value in a $\PREPARE$ message can include up to $n$ distinct requests, and a value in a $\ACCEPT$ message may have arbitrary size.
Therefore, in Appendix~\ref{app:refinedAlg}, we present a refined protocol description in which processes exchange $O(n^2)$ messages \emph{per individual request}.
This efficiency is achieved by relaying only the differences between current and previously received proposals and the learned values in phases $2$ and $3$, thus eliminating redundant messages with the same requests.

\begin{algorithm}
{\footnotesize

\begin{smartalgorithmic}[1]

\Upon{\Startup}
    \State $\MPool, \Proposing, \Validated, \Learned \gets \perp$
    \State $\Pending \gets \emptyset$
\EndUpon
\algspace[0.5]

\Operation{\Propose}{$v$}
    \State \SendRequest($v$)
    \State wait until $v \sqsubseteq \Learned$ 
    \State \Return $\Learned$
\EndOperation
\algspace[0.5]

\Operation{\SendRequest}{$v$}
    \State $\MPool \gets \MPool \sqcup v$\label{line:mpool}
    \State \Send{$\REQUEST, v$}{every other node} \label{line:sendReq1}
\EndOperation
\algspace[0.5]

\UponReceiving{$\REQUEST, v$}{a node}
    \If{$v \not\sqsubseteq \MPool \sqcup \Proposing \sqcup \Learned$}
        \State $\MPool \gets \MPool \sqcup v$
        \State \Send{$\REQUEST, v$}{every other node} \label{line:sendReq2}
    \EndIf
\EndUponReceiving
\algspace[0.5]

\UponEvent{$(\MPool \neq \perp) \wedge (\Proposing  = \perp)$}\label{line:startProp}
    \State $\Proposing \gets \MPool$
    \State $\MPool \gets \perp$
    \State $\Pending[\Proposing] \gets 1$ \label{line:pending}
    \State \Send{$\PROPOSE, \Proposing$}{every other node}\label{line:propose1} \label{line:sendProp1}
\EndUponEvent
\algspace[0.5]

\UponReceiving{$\PROPOSE, v$}{a node}
    \If{$v \in \Pending.\keys()$}
        \State $\Pending[v]++$
    \Else
        \State $\Pending[v] \gets 1$
        \State \Send{$\PROPOSE, v$}{every node}\label{line:prepM}
    \EndIf
\EndUponReceiving
\algspace[0.5]

\Upon{exists $v$ s.t. $\Pending[v]=n-f$}\label{line:valValue}
    \State $\Validated \gets \Validated \sqcup v$
\EndUpon
\algspace[0.5]

\UponEvent{$\bigsqcup\Pending.\keys() \sqsubseteq \Validated$}\label{line:learnCond1}
    \If{$\Learned \sqsubset \Validated$}
        \State $\Learned \gets \Validated$
        \State $\Proposing \gets \perp$ \label{line:endProp1}
        \State \Send{$\ACCEPT,\Learned$}{every node} \label{line:sendAcc1}
    \EndIf
\EndUponEvent
\algspace[0.5]

\UponReceiving{$\ACCEPT, w$}{a node}
    \label{line:learnCond2}\If{$(\Proposing \sqcup \Learned \sqsubseteq w)$}
        \State $\Validated \gets \Validated \sqcup w$
        \State $\Learned \gets w$
        \State $\Proposing \gets \perp$ \label{line:endProp2}
        \State \Send{$\ACCEPT,\Learned$}{every node} \label{line:endLearnCond2}
    \EndIf
\EndUponReceiving

\end{smartalgorithmic}
}
\caption{Long-Lived LA: code for node $x$.}
\label{alg:longLA}
\end{algorithm}

\subsection{Correctness}

$\GValidity$ and $\GStability$ are immediate.
We now proceed with $\GConsistency$ and $\GLiveness$.  

\begin{lemma}
\label{lm:comparability}
If nodes $i$ and $j$ learn, resp., values $w_i$ and $w_j$, then $w_i$ and $w_j$ are comparable.

\begin{proof}
Suppose that $(w_i \not\sqsubseteq w_j) \wedge (w_j \not\sqsubseteq w_i)$. 
Then there must exist $v_i \sqsubseteq w_i$ and $v_j \sqsubseteq w_j$ such that $v_i \not\sqsubseteq w_j$ and $v_j \not\sqsubseteq w_i$.

Let $Q_i$ (resp. $Q_j$) be the quorum $i$ used to include $v_i$ $\Validated$ at line~\ref{line:valValue}.
Since $Q_i \cap Q_j \neq \emptyset$, there is a common node $x$ that sent $\langle \PREPARE, v_i \rangle$ to $i$ and $\langle \PREPARE, v_j \rangle$ to $j$, but since channels are FIFO, either $i$ received $v_j$ or $j$ received $v_i$ from $x$ before learning a value, therefore adding the value to $\Pending$.
Suppose it was $i$ that received $v_j$ before $v_i$, from the condition of line~\ref{line:learnCond1}, $i$ could not have learned $w_i$ if $v_j \not\sqsubseteq \Validated$.
\end{proof}
\end{lemma}

\begin{lemma}
\label{lm:replicaLiv}
If a correct node $x$ sets $\Proposing = v$, $x$ eventually learns a value with $v$. 

\begin{proof}
A node $x$ sends a $\PREPARE$ message to every other node whenever it adds a new value to $\Pending$ (line~\ref{line:prepM}).
If $x$ is correct, it will receive at least $n-f$ $\PREPARE$ messages for every value in $\Pending$, adding the value to $\Validated$.
Therefore, the condition in line~\ref{line:learnCond1} is never  satisfied from some point on only if $x$ keeps adding a new value to $\Pending$ before all the current ones are validated.

Since each node proposes only one value at a time (until it learns a value, lines~\ref{line:startProp},~\ref{line:endProp1},~\ref{line:endProp2}), for $x$ to indefinitely add new values to $\Pending$, there must be at least one other node that keeps learning values and proposing new ones.
Without loss of generality, let $y$ be one such node. 
Since faulty nodes eventually crash and stop taking steps, $y$ must be correct.
Every time $y$ learns a new value $w$ it sends $\langle \ACCEPT, w \rangle$ to $x$, and because channels are FIFO, $x$ receives the $\ACCEPT$ message before the new value proposed by $y$.
Eventually (because $x$ sent its proposal to $y$), one of the received values $w$ contains $x's$ $\Proposing$ and the condition on line~\ref{line:learnCond2} is satisfied, $x$ then learns $w$.
\end{proof}
\end{lemma}

\begin{lemma}
\label{lm:callProp}
    If a correct node calls $\Propose(v)$, it eventually sets $\Proposing = v'$, $v \sqsubseteq v'$.
\end{lemma}

\begin{proof}
    Let a correct node $x$ call $\Propose(v)$, $x$ then includes $v$ in $\MPool$ (line~\ref{line:mpool}).
    If $x$ is not currently proposing, that is, the current value of $\Proposing$ is $\perp$, then it meets the condition in line~\ref{line:startProp} and immediately sets $\Proposing = \MPool$.
    Otherwise, by Lemma~\ref{lm:replicaLiv}, it eventually learns a value and sets $\Proposing = \perp$ in lines~\ref{line:endProp1} and~\ref{line:endProp2}, thus meeting the condition in line~\ref{line:startProp} and setting $\Proposing = \MPool$. 
\end{proof}

Lemmas~\ref{lm:comparability},~\ref{lm:replicaLiv} and~\ref{lm:callProp} imply:

\begin{theorem}
    Algorithm~\ref{alg:longLA} implements Generalized Lattice Agreement.
\end{theorem}

\begin{corollary}
    Algorithms~\ref{alg:GLAtoAS} and~\ref{alg:longLA} implement Atomic Snapshot.
\end{corollary}

\subsection{Time metric}
\label{subsec:ourMetric}

We now define the latency metric we are going to use in evaluating time complexity.  
Our metric is inspired by the metric proposed by Abraham et al.~\cite{abraham-metric} (which in turn rephrases the original metric by Canetti and Rabin~\cite{canetti1993fast}).
The distinguishing feature of our approach is that 
it also applies to \emph{long-lived} executions and executions with \emph{holes} (illustrated in Figure~\ref{fig:egIRA}).\footnote{In Section~\ref{sec:timeMeasure}, we show that the three metrics are equivalent in "hole-free" executions.}

Algorithm~\ref{alg:refinedNTR} describes the iterative method that assigns rounds to events in an execution.
We give an informal description of the metric below.

\begin{definition}[\textbf{I}terative \textbf{R}ound \textbf{A}ssignment - Informal]
Algorithm~\ref{alg:refinedNTR} assigns round $0$ to the initial event, and defines the end of round $i$ as the last event that receives a message sent in round $i-1$.
In addition, if there are no more messages to be received (or in transit), the event inherits the round number of its immediate predecessor.
\end{definition}

\begin{algorithm}
{\footnotesize
\begin{smartalgorithmic}[1]

\State $e_0^* := e_0$
\State $e_0$ is assigned round $0$
\State $r := 0$

\For{i=1...}
    \If{$e_i$ does not receive a message}
        \State $e_i$ is assigned round $r$ \label{line:momentRound}
    \Else
        \State Let $e_j$ be the oldest event from which $e_i$ receives a message
        
        \State Let $r'$ be the round assigned to $e_j$ $(r' \leq r)$
        \State Let $e'$ be the most recent event among $e_{r'}^*$ and $e_j$
        \State All events after $e'$ and up to $e_i$ receive round $r'+1$ \label{line:defRound}
        \State $e_{r'+1}^* := e_i$ \label{line:lastStar}
        \State $r = r'+1$
    \EndIf    
\EndFor

\end{smartalgorithmic}
}
\caption{Iterative Round Assignment (\IRA)}
\label{alg:refinedNTR}
\end{algorithm}

\begin{figure}[tp]
    \centering
    \includegraphics[width=0.5\textwidth]{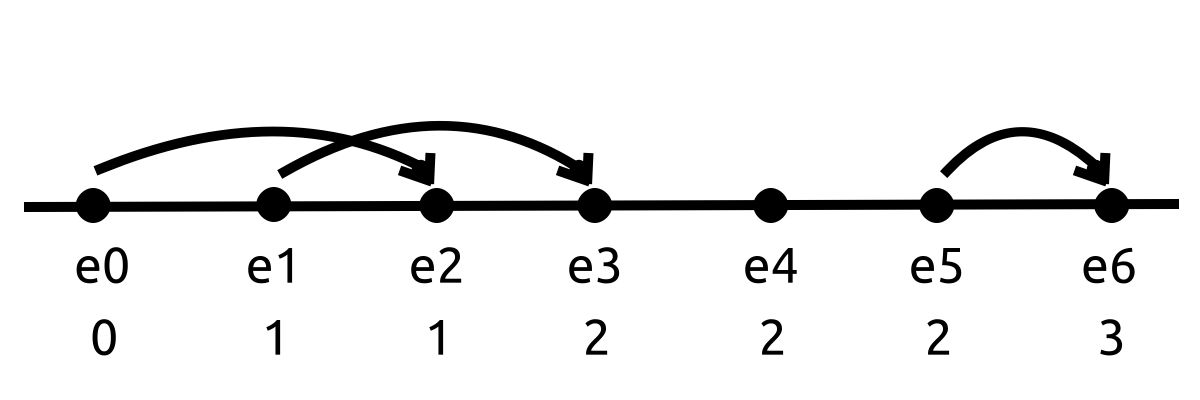}
    \caption{Example of round assignment using $\IRA$. Arrows represent message transmissions and the number below an event corresponds to its round. A ``hole'' in communication appears betwen events $e_3$ and $e_5$.}
    \label{fig:egIRA}
\end{figure}

\begin{definition}[$\IRA$ - Arbitrary Events]
\label{def:iraEvents}
To measure the \emph{latency between two events $e_i$ and $e_j$}, we assign rounds according to Algorithm~\ref{alg:refinedNTR}, \emph{starting from} $e_i$, with \emph{all events} up to and including $e_i$ receiving round $0$.
The latency between $e_i$ and $e_j$ is then given by the round assigned to $e_j$.
\end{definition}

We say that an application request (or simply request, when there is no ambiguity) \emph{completes} once the receiving node learns a value which includes the request.
For a specific node $i$, we are interested in measuring the latency between the event $e_C$ in which $i$ receives a value $v$ from the application software, and an event $e_R$, in which $i$ learns a value $w$ with $v$.

\subsection{Time complexity of Algorithm~\ref{alg:longLA}}
We define latency as the number of rounds spanning between the moment a correct process receives an application call and the moment it returns from the operation.
In evaluating the latency of our protocol, we consider two types of executions: (1)~the \emph{fault-free} case, when all processes are correct, and (2)~the \emph{worst-case}, when only a majority of processes are correct.

A snapshot operation $\Op$ \emph{precedes} another operation $\Op'$ if the response event of $\Op$ happens before the call event for $\Op'$.
Two operations are said to be \emph{concurrent} if none precedes the other.
For $\ASO$ protocols, we analyze latency in fault-free runs of an operation $\Op$ in two distinct scenarios: (a)~without contention, i.e., when no other operation overlaps in time with $\Op$, and (b)~with contention, i.e., when there might be an arbitrary number of concurrent operations.

Garg et al.~\cite{ASOconstant} use the notion of \emph{amortized time complexity}, i.e., the average operation latency taken over a large number of operations in an execution.
In some protocols, including ours, the latency of an operation is only affected by the number of faulty processes whose messages are received during the operation's interval (we call these processes \emph{active-faulty}).
Intuitively, faulty processes take a finite amount of steps, so in these protocols a failure can only affect a finite number of operations.
In this paper, we also distinguish $\ASO$ protocols with \emph{constant} time complexity.

Next, we establish the optimality of our protocol under no-contention.
A protocol implementing $\LA$ \emph{tolerates} $k$ faults if it satisfies all the properties of Definition~\ref{def:LA} in every execution with at most $k$ faulty processes.

\begin{theorem}
    Let $\cP$ be a distributed protocol that implements $\LA$ and tolerates at least one faulty process. Then, there exists a fault-free run of $\cP$ in which an $\LA$ operation requires at least two rounds of communication to complete without contention.
\end{theorem}

\begin{proof}
Consider an operation $\Op$ initiated by node $x$, with call event $e_C$ and response event $e_R$. Suppose $\Op$ completes in at most one round in fault-free, contention-free executions.

We first show that there exists an execution $E = e_1, \ldots, e_C, \ldots, e_R$ such that:
\begin{itemize}
    \item $x$ is the only process to take a step in $e_R$,
    \item no message sent by $x$ in $e_C, \ldots, e_R$ is received by any other process before $e_R$.
\end{itemize}

If multiple processes perform steps in the same event $e$, we can conceptually "split" $e$ into a sequence of events $e^1, e^2, \ldots$, where each process takes the step in its own dedicated event.
Since their steps are independent, these split events are indistinguishable from the original $e$ from each process's perspective.
This reasoning also applies to $e_R$.

Now, assume for the sake of contradiction that in every fault-free, contention-free execution containing both $e_C$ and $e_R$, there exists some process $y \neq x$ that receives a message $m$--sent by $x$ in the interval $e_C, \ldots, e_R$--before $e_R$ occurs.

Let $e_M$ denote the event where $y$ receives $m$. We define rounds from $e_C$'s perspective:
\begin{itemize}
    \item all events up to $e_C$ are in round $0$,
    \item round $1$ ends at the last event $e_L$ that receives a message originating in round $0$.
\end{itemize}

If $m$ is sent after $e_C$, then we can construct $E$ so that all messages from round $0$ are received before $m$.
This ensures that $e_M$ occurs after $e_L$, meaning $e_M$ is in round 2.
Since $e_R$ occurs after $e_M$, it too is assigned round 2--contradicting our assumption that $\Op$ completes in one round.

If instead $m$ is sent in $e_C$, we can again construct the execution so that all round $0$ messages are received before or at the same time as $m$, making $e_M = e_L$.
Since $e_R$ occurs after $e_M$, it is again assigned to round 2--a contradiction.

These contradictions hold regardless of whether $\Op$ is concurrent with any other operation.
Hence, such an execution $E$ must exist.
Now consider an extension $E'$ of $E$ where all messages sent by $x$ after (and including) $e_C$ are indefinitely delayed, while messages from other nodes are not.

Suppose a node $z$ invokes a new operation $\Op'$ after $e_R$, making $\Op'$ non-concurrent with $\Op$.
Since protocol $\cP$ tolerates at least one faulty process, and $x$ appears to have crashed in $E'$, node $z$ must eventually complete $\Op'$ without any process receiving any messages from $x$.

Let $v$ and $w$ be the value proposed and the value learned by $x$ in $\Op$, and let $v'$ and $w'$ be the corresponding values for $z$ in $\Op'$. By $\GValidity$, we know $v \sqsubseteq w$, and by $\GConsistency$, we know $w \sqsubseteq w'$, hence $v \sqsubseteq w'$.

However, since no process receives a message from $x$ since $e_C$, no one could have known about $v$, contradicting the requirement that $w'$ must contain $v$.

Finally, after $\Op'$ completes, we can allow all delayed messages from $x$ to be received, making all processes correct in the final execution $E'$. This completes the proof.
\end{proof}

\begin{theorem}
   In a fault-free run without contention, a request takes at most $2$ rounds to complete.
\end{theorem}
   \begin{proof}
    Consider a contention-free request with call event $e_C$ and return event $e_R$ invoked by a node $i$. 
    There are no call events for other nodes between $e_C$ and $e_R$, but some messages from previous proposals may still be in transit.

    Suppose $v$ is the value to be proposed for the application call.
    If $i$ is not proposing (has $\Proposing = \perp$) when it receives $v$, then it directly sends $\langle \PREPARE, v \rangle$ to everyone. Let $e_P$ be the last event in which a process receives $\langle \PREPARE, v \rangle$ from $i$, then every process also sends $\langle \PREPARE, v \rangle$ by at most $e_P$.
    Now take $e_F$ as the final event in which a process receives $\langle \PREPARE, v \rangle$ in the execution, and $e_S$ as the corresponding sending event.
    It must be that $e_S$ happens between $e_C$ and (potentially including) $e_P$.
    Also, because the channels are FIFO, every previous proposal must have been validated before $e_F$, and $i$ will learn a value containing $v$ by at most $e_F$. 
    Let $e_C$ be assigned round $0$, then $e_P$ happens at most in round $1$.
    As a consequence, $e_S$ is assigned either $0$ or $1$, thus $e_F$ can be assigned at most round $2$. 
    Then, by the end of round $2$, $i$ already has $v$ validated.

    Now suppose that $i$ is proposing when it receives $v$, so it still has a value $v'$ in $\Pending$ that is not validated, w.l.o.g. assume that $v'$ is the only one.
    This value must be from a call that already finished, and the corresponding node sent $\langle \ACCEPT, w \rangle$ containing $v'$ before $e_C$.
    Consider two pairs of events: $(e_A,e_A')$ and $(e_C,e_C')$.
    In the first pair, $e_A$ is the event where $\langle \ACCEPT, w \rangle$ was first sent, and $e_A'$ is the \emph{last} event in which $\langle \ACCEPT, w \rangle$ is received from $e_A$.
    In the second, $e_C$ is the usual application call event and $e_C'$ is the \emph{last} event in which $\langle \REQUEST, v \rangle$ is received from $i$.
    There are two cases to consider: 1) $e_A'$ happens before $e_C'$ and 2) $e_C'$ happens before $e_A'$.

    If it is the first case, then at the moment $e_C'$ happens, every node was already able to propose $v$ (since there was no other value to be learned).
    Take the last event $e_L$ in which a $\langle \PREPARE, v \rangle$ (or a value containing $v$) is received, and $e_S$ as the corresponding sending event, it follows that $i$ validates $v$ by at most $e_L$ and can learn a value containing it.
    Let $e_C$ be assigned round $0$, $e_C'$ and $e_S$ can be assigned at most round $1$, and since $e_L$ receives a message from $e_S$, it can be assigned at most round $2$.
    If it is the second case, then all nodes received $\langle \REQUEST, v \rangle$ and put $v$ in $\MPool$ before $e_A'$.
    Every node proposes $v$ by at most $e_A'$ (since they can adopt $w$ and stop any current proposal).
    Let $e_L$ be the last event in which a process receives a proposal for $v$ and $e_S$ it's corresponding sending event, similarly to the above cases, $e_S$ happens between $e_C$ and $e_A'$.
    Now, let $e_A$ and $e_C$ be assigned round $0$.
    $e_S$ can be assigned at most round $1$ ($e_S$ happens before or at $e_A'$) and $e_L$ at most $2$, which concludes the proof.
   \end{proof}

Consider an execution of our algorithm, and let $F$ ($|F| \leq f$) be its set of faulty processes.

\begin{lemma}
\label{lm:timeLearn}
    Consider an event in which a correct node sends $\langle \PREPARE, v \rangle$ and the first event in which a correct node learns a value including $v$.
    If no correct node receives a message from a faulty one between these two events, then there are at most $3$ rounds between them.
\end{lemma}

\begin{proof}
    A message sent by a correct node is received by every correct node in the execution, and since correct nodes do not receive messages from faulty ones in the interval we are analyzing, we can consider only events originated from correct nodes.
    Therefore, we only refer to correct nodes in the following.
    
    Let $x$ be the node sending $\langle \PREPARE, v \rangle$, $e_P$ be the corresponding event and $e_P'$ the last event a node receives $\langle \PREPARE, v \rangle$ from $x$.
    Because $x$ also sends $\langle \REQUEST, v \rangle$, by $e_P'$ every node received the request and must be proposing.
    Any value learned after $e_P'$ contains $v$ since all nodes have $v$ in $\Pending$.

    Now, at the configuration just after applying $e_P'$, let $V$ be the set in which $w \in V$ satisfies: there exists a (correct) node where $w$ is in $\Pending$ but is not yet validated.
    Consider a value $w \in V$ that is the last whose $\langle \PREPARE, w \rangle$ is received by any node, where $e_L'$ is the event in which $\langle \PREPARE, w \rangle$ is last received and $e_L$ the corresponding sending event.
    It follows that some node learns a value containing $v$ by at most $e_L'$.

    Next, take the first event $e_F$ in which a node sent $\langle \PREPARE, w \rangle$, and $e_F'$ the event in which the last $\langle \PREPARE, w \rangle$ from $e_F$ is received.
    Note that $e_L$ happens at most at $e_F'$ and $e_F$ at most at $e_P'$.
    Let $e_P$ be assigned round $0$, then $e_P'$ (and thus $e_F$) can be assigned at most round $1$, $e_F'$ (and thus $e_L$) at most $2$ and lastly, $e_L'$ can be assigned at most round $3$. Therefore, there are at most $3$ rounds between a propose and the first learn event for $v$.
\end{proof}

\begin{theorem}
    \label{th:8rounds}
    An operation $\Op$ takes at most $8$ rounds to complete if, during its interval, no correct node receives a message from a faulty one.
\end{theorem}

\begin{proof}
    Let $v$ be the value received from the application call for $\Op$, $e$ be the event in which node $i$ proposes $v$ (or a value containing $v$) and $e'$ the event in which a value including $v$ is learned for the first time.
    From Lemma~\ref{lm:timeLearn}, there are at most $3$ rounds between $e$ and $e'$.
    Since the node that learns $v$ sends $\langle \ACCEPT, v \rangle$ to everyone, $i$ receives and adopts it in one extra round.
    We conclude that in at most $4$ rounds every correct node can learn $v$.

    If $i$ is already proposing a value when it receives a call for $v$, it sends $\langle \REQUEST, v \rangle$ to everyone and put it in $\MPool$, so it is proposed next.
    Let $e_P$ be the event in which $i$ initiated its previous proposal to $v$, and consider the worst case where the application call $e_C$ with $v$ happens just after $e_P$.
    From $e_P$ to the event in which $i$ learns its previous proposal $e_P'$ (and thus starts proposing $v$), there are at most $4$ rounds, and from $e_P'$ to the learning event of $v$ there are also at most $4$ rounds.
    Therefore, the operation completes in at most $8$ rounds.
\end{proof}

We say that there are \emph{$k$ active faulty nodes during an operation $\Op$} if, in between the call and return events for $\Op$, a message is received from a total of $k$ distinct faulty nodes.

\begin{theorem}
\label{th:ok}
An operation $\Op$ takes $O(k)$ rounds to complete, where $k$ is the number of active faulty nodes during $\Op$.
\end{theorem}

\begin{proof}
    See Appendix~\ref{app:ourLAcomp}.
\end{proof}

\begin{corollary}
    Algorithms~\ref{alg:GLAtoAS} and~\ref{alg:longLA} together have an amortized time complexity of $8$ rounds.
\end{corollary}

\section{Measuring latency of ASO protocols}
\label{sec:measuringASO}

We conclude the paper with an overview of time complexity of earlier {\LA} and {\ASO} protocols~\cite{faleiro2012podc,delporte2018implementing,imbs2018set,garg2020amortized,ASOconstant}.
We highlight certain gaps in their latency analyses and discuss the ways to fix them. 
Formalities and proofs are delegated to the appendix. 

\myparagraph{The first message-passing $\LA$ protocol}
Faleiro et al.~\cite{faleiro2012podc} came up with the first $\LA$ implementation for asynchronous message-passing systems.
They use the metric of~\cite{ABDpaper} to measure latency and conclude that it takes $O(n)$ rounds to output from a lattice agreement operation in their protocol.

We show in Appendix~\ref{app:ASO} the somewhat surprising result that  this protocol has constant latency of $16$ rounds in fault-free runs.
The upper bound holds as long as no message from faulty processes is received during the interval of the operation, implying that their $\LA$ protocol has constant amortized time complexity.
We conjecture that the protocol has $O(k)$ worst-case latency, where $k$ is the number of actual failures in the execution.

\myparagraph{The first direct $\ASO$ implementation}
Delporte et al.~\cite{delporte2018implementing} is the first paper to directly implement $\ASO$ in message passing systems, instead of using an atomic register implementation~\cite{ABDpaper} and the shared-memory snapshot construction~\cite{atomic-snapshot}.

In fault-free runs without contention, the latency of their protocol is only $2$ rounds.
In fault-free runs with contention, we support the claim of a bound of $O(n)$ rounds from~\cite{ASOconstant}.

\myparagraph{ASO with SCD-Broadcast}
Imbs et al.~\cite{imbs2018set} introduce the abstraction of \emph{Set Constrained Delivery Broadcast} ($\SCD-\Broadcast$),
and show that it allows for implementing $\LA$ and $\ASO$ with no complexity overhead.
In their complexity analysis, they assume bounded message delays and show that the latency of their $\ASO$ algorithm \emph{in faulty-free and contention-free runs} is $2$ rounds.
In Appendix~\ref{app:ASO}, we show that an operation of their resulting {\ASO} algorithm can take $\Omega(n)$ rounds in fault-free runs \emph{with contention}.
We conjecture that this bound is tight, and so the time complexity of their $\ASO$ protocol is $\Theta(n)$.

\myparagraph{A generic $\ASO$ algorithm}
Garg et al.~\cite{garg2020amortized,ASOconstant} give a generic construction for atomic snapshot which uses any one-shot $\LA$ protocol (see definition in Appendix~\ref{app:oneShot}) as a building block (with constant latency overhead).
The protocol thus inherits the asymptotic complexity of the underlying $\LA$ algorithm.
They also provide a protocol for one-shot $\LA$ with $2$ rounds latency in fault-free runs (using ~\cite{canetti1993fast}'s metric).
Their protocol requires 2 rounds of communication plus two lattice agreement invocations in the good case w/o contention and three lattice invocations with contention, making it at least 6 and 8 message delays, respectively.

For the worst-case latency analysis, they assume an additional requirement over communication channels: if a process executes $\send(m)$, sending $m$ to a \emph{correct} process, then $m$ is eventually received (even if the sender is faulty).
Using this assumption, they show a worst-case latency of $O(\sqrt{f})$ for their $\LA$ protocol.

In this paper, we assume a weaker channel that only guarantees delivery of messages among correct processes.
We show that under this model, the $\LA$ protocol of~\cite{garg2020amortized} has an execution that takes $\Omega(f)$ rounds.
We conjecture the upper bound of their protocol to be $O(f)$, and also that when using the stronger assumption, both our (Section~\ref{sec:ourProtocol}) and~\cite{imbs2018set}'s protocol have $O(\sqrt{f})$ worst-case latency.

The generic $\ASO$ construction may also be combined with the one-shot $\LA$ protocol presented in~\cite{zheng2019opodis}, which has worst-case latency of $O(\log f)$, providing an object whose update and snapshot operations take $O(\log f)$ in both fault-free fault-prone executions.
For the sake of completeness, we also provide the time complexity analysis for the one-shot $\LA$ protocols from~\cite{faleiro2012podc} and~\cite{imbs2018set} in Appendix~\ref{app:oneShot}.

\section{Comparative Analysis of Time Measurement Metrics}
\label{sec:timeMeasure}

In this section, we recall metrics used in the literature~\cite{ABDpaper,canetti1993fast,abraham-metric,lamport1978time} for measuring time in asynchronous systems.
We exhibit executions where the metrics by Attiya et al.~\cite{ABDpaper} and Canetti and Rabin~\cite{canetti1993fast} yield arbitrary results due to the presence of \emph{holes} -- ``periods of silence'' during which no messages are in transit -- which are common in long-lived protocols.
We show that in a subset of executions without holes, which we refer to as \emph{covered executions}, these metrics align with the one proposed by Abraham et al.~\cite{abraham-metric}.
This is not surprising, as these metrics were designed for \emph{distributed tasks}, which assume finite hole-free executions.
We also recall Lamport’s longest causal chain metric~\cite{lamport1978time} and show that it is not suitable for comparing the $\ASO$ protocols we consider here.

Next, we show that the metric from~\cite{abraham-metric} diverges from~\cite{ABDpaper} and~\cite{canetti1993fast} when na\"ively applied to measure time between arbitrary events.
We then show that, after employing our refined method from Section~\ref{subsec:ourMetric}, they match when measuring rounds between arbitrary events in covered executions.

Finally, we show that both our metric and that of~\cite{abraham-metric} yield equivalent results in cases where~\cite{abraham-metric} is applicable.
Altogether, we establish that our metric generalizes~\cite{abraham-metric} and aligns with classical metrics~\cite{ABDpaper,canetti1993fast} when applied to distributed tasks.
A summary of the comparative analysis is presented in Table~\ref{tab:metrics}.

\begin{table*}[tp]
     \begin{center}
  \begin{tabular}{||c|c|c|c|c||}
  \hline
   & \makecell[c]{Timed} & \makecell[c]{Equivalent to $\CR$ \\ (Covered Executions)} & \makecell[c]{Equivalent to $\CR$ \\ (Arbitrary Events)} & \makecell[c]{Admits \\ Holes}\\
  \hline
 
  \hline
 $\CR$~\cite{canetti1993fast} & \textcolor{red}{Yes} & - & - & \textcolor{red}{No} \\ \hline
 $\Round$~\cite{ABDpaper,attiya2004distributed} & \textcolor{red}{Yes}  & \textcolor{blue}{Yes} & \textcolor{blue}{Yes} & \textcolor{red}{No} \\ \hline
 
 $\NTR$~\cite{abraham-metric} & \textcolor{blue}{No}  & \textcolor{blue}{Yes} & \textcolor{olive}{Yes}
 & \textcolor{red}{No} \\ \hline
 
 $\LCC$~\cite{lamport2006lower} & \textcolor{blue}{No} & \textcolor{red}{No} & \textcolor{red}{No} & \textcolor{blue}{Yes} \\ \hline
 
 $\IRA$ & \textcolor{blue}{No} & \textcolor{blue}{Yes} & \textcolor{blue}{Yes} & \textcolor{blue}{Yes} \\ \hline
 \end{tabular}
 \end{center}
     \caption{Comparison between asynchronous time metrics. Metrics that are \emph{timed} make use of time assignments to determine the number of rounds between events. We compare each metric against $\CR$, evaluating the number of rounds resulting from applying them over entire (covered) executions and between arbitrary events. \textcolor{blue}{Blue} stands for "good" features and \textcolor{red}{red}---for "bad" ones. \textcolor{olive}{The equivalence of NTR to CR} holds as long as one uses Definition~\ref{def:iraEvents}.}
     \label{tab:metrics}
\end{table*} 

\subsection{Definitions}
\myparagraph{Timed Executions} We assume a global clock, not accessible to the nodes. 
A \emph{timed event} $\overline{e}$ is a pair $(t,e)$ in which $t$ is a non-negative real number, we also say that $\overline{e}$ is a \emph{time assignment} of $e$.
A \emph{timed execution} is an alternating sequence $C_0\overline{e}_1C_1\dots$ where $\overline{e}_1 = (t_1, e_1),\overline{e}_2 = (t_2, e_2),\dots$, where events $e_1,e_2,\ldots$ are equipped with monotonically increasing times $t_1,t_2,\ldots$: 
\begin{enumerate}
    \item $t_m > t_l$ whenever $m > l$;
    \item $t_l \rightarrow \infty$ as $l \rightarrow \infty$.\footnote{We require this property to avoid the case where a never-terminating execution has a finite time duration.}
\end{enumerate}

A \emph{time assignment} of $E$ is a timed execution $\overline{E}$ in which every event $e_i$ in $E$ is matched with a timed event $(t_i,e_i)$ in $\overline{E}$ and the sequences of configurations in $E$ and $\overline{E}$ are the same.
Notice that an execution allows for infinitely many time assignments.

Let $m$ be a message sent in $\overline{e}_l$ and received in $\overline{e}_m$, the \emph{delay} of $m$ is then defined as $t_m - t_l$.
For a finite timed execution 
$\overline{E} = C_0\overline{e}_1...\overline{e}_l C_l$, we define 
$t_{\Start}(\overline{E}) = t_1$, 
$t_{\End}(\overline{E}) = t_l$ (we use $t_{\Start}$ and $t_{\End}$ 
when there is no ambiguity) and $\Duration(\overline{E}) = t_{\End} - t_{\Start}$.

In the subsequent discussion, given an execution $E$, let $\cT(E)$ denote the set of all timed executions $\overline{E}$ based on $E$.

\myparagraph{Time Metrics}
It is conventional to measure the execution time by the number of communication \emph{rounds}, typically calculated using the ``longest message delay.'' These metrics can be applied to both \emph{executions} and \emph{timed executions}.
The first metric we consider is defined in \Cref{def:roundMetric}~\cite{ABDpaper}. 
When applied to timed executions, this metric assumes a known upper bound on message delays, which can be normalized to one time unit without loss of generality. 
To apply this metric to an \emph{execution}, we consider the maximum duration of all possible timed executions that adhere to the upper-bound communication constraint.

\begin{definition}[Round metric]
\label{def:roundMetric}
Given a timed execution $\overline{E}$, in which the maximum message delay is bounded by one unit of time, $\overline{E}$ \emph{takes $\Duration(\overline{E})$ rounds}.

By extension, an execution $E$ takes $\sup_{\overline{E} \in \cT(E)}{\Duration(\overline{E})}$ rounds.
\end{definition}

In the metric proposed by Attiya and Welch~\cite{attiya2004distributed, attiya2024multi}, the time assignments are scaled so that the maximum message delay is always $1$, thus, the metric produces the same results for executions as Definition~\ref{def:roundMetric}.
A more general metric introduced by Canetti and Rabin~\cite{canetti1993fast} captures the time complexity of any finite execution. 
Let $\overline{E}$ be a timed execution, and let $\delta_{\overline{E}}$ be the maximum message delay in it.
Then $\overline{E}$ \emph{takes $\Duration(\overline{E})/\delta_{\overline{E}}$ \emph{\CR} rounds}.

\begin{definition}[$\CR$ metric]
\label{def:crMetric}
A finite execution $E$ takes $\sup_{\overline{E} \in \cT(E)}{\Duration(\overline{E})/\delta_{\overline{E}}}$ rounds, where $\delta_{\overline{E}}$ is the maximum message delay of each corresponding timed execution.
\end{definition}

\begin{figure}[tp]
    \centering
    \includegraphics[width=0.5\textwidth]{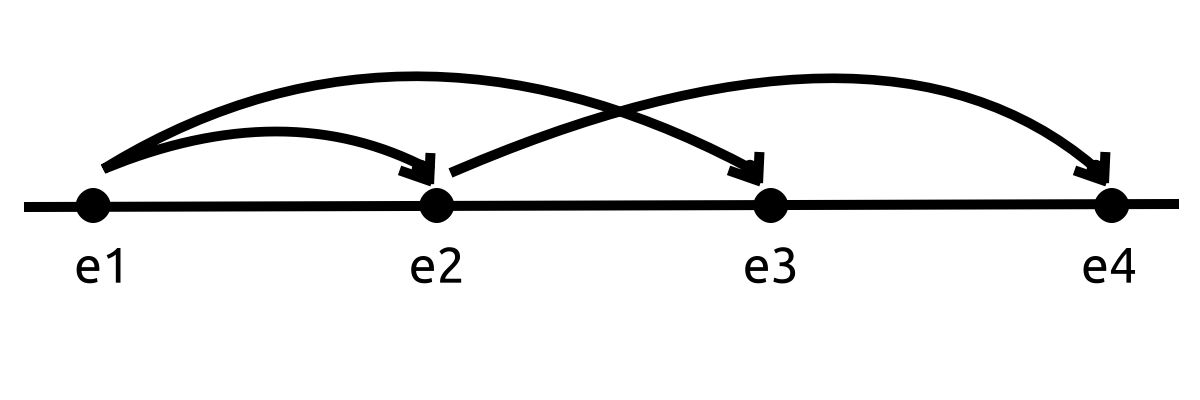}
    \caption{Example of an execution with 2 rounds in the Round, CR and NTR metrics.}
    \label{fig:2cr}
\end{figure}

\begin{example}

Figure~\ref{fig:2cr} shows an execution with four events, where
we assign a delay of $\delta$ to the message exchanges $(e_1,e_3)$ and $(e_2,e_4)$, and a delay of $\delta - \epsilon$ ($\epsilon > 0$) to $(e_1,e_2)$.
By making $\epsilon$ arbitrarily small, the number of rounds in this execution converges to $2$ in the $\CR$ metric. 
The same result is obtained in the Round metric by setting $\delta = 1$.
\end{example}

Recently, Abraham et al.~\cite{abraham-metric} proposed an elegant approach that can be directly applied to executions without relying on time assignments. 
We call this metric \textbf{n}on-\textbf{t}imed \textbf{r}ounds ($\NTR$):

\begin{definition}[$\NTR$ metric]
\label{def:nonTimeRounds}
Given an execution $E$, each event in $E$ is assigned a round number as follows:
\begin{itemize}
    \item The first event $e_0$ is assigned round $0$. We also write $e_0^* = e_0$;
    \item For any $r \geq 1$, let $e_r^*$ be the last event where a message of round $r-1$ is delivered. All events after $e_{r-1}^*$ until (and including) event $e_r^*$ are in round $r$.
\end{itemize}
The number of rounds in $E$ is the round assigned to its last event.
\end{definition}

\begin{example}
    Coming back to Figure~\ref{fig:2cr}, if we assign a round to each event based on Definition~\ref{def:nonTimeRounds} then $e_1$ gets round $0$, $e_2$ and $e_3$ get round $1$ and $e_4$ is assigned round $2$. The execution has therefore $2$ rounds according to $\NTR$.
\end{example}

Lamport~\cite{lamport2006lower} proposed a metric for latency based on the causal chain of messages.
The \emph{Longest Causal Chain} ($\LCC$) was used to show best-case latency of protocols such as consensus~\cite{lamport2006lower} and Crusader Agreement~\cite{abraham2023round}.

\begin{definition}
    [\textbf{L}ongest \textbf{C}ausal \textbf{C}hain] Let $e$ be an event in $E$ and $M$ the set of messages received by $e$, then $e$ is assigned round $k + 1$, where $k$ is the maximum round of an event originating a message in $M$. If $M = \emptyset$, then $k = 0$.
    The number of rounds in an execution becomes the highest round assigned to one of its events.
\end{definition}

This metric, however, diverges from $\CR$ and $\NTR$.

\begin{example} [Reliable Broadcast]
\label{ex:broadcast}
In the \emph{reliable broadcast} primitive~\cite{cachin2011introduction}, a dedicated source \emph{broadcasts} a message and, if the source is correct, then all correct nodes should deliver the message.
Furthermore, if a correct process delivers a message, then every correct process eventually delivers it.
The following protocol satisfies this property:

\begin{itemize}
    \item When the source invokes \emph{broadcast($m$)}, it delivers $m$ and sends it to everyone;

\item When a process receives $m$ for the first time, it delivers $m$ and sends it to everyone.
\end{itemize}

In Figure~\ref{fig:broadcast}, we depict an execution of this protocol with four processes: $p_1$, $p_2$, $p_3$ and $p_4$.
Here, $p_1$ is the source and broadcasts $m$, the message is received by $p_2$ which then sends $m$ to everyone.
Process $p_3$ receives $m$ from $p_2$ before receiving it from $p_1$, and finally, $p_4$ receives $m$ from $p_1$ in the last event.
This execution has $2$ $\LCC$ rounds, while having $1$ round according to $\CR$ and $\NTR$.
\end{example}

\begin{figure}[tp]
    \centering
    \includegraphics[width=0.5\textwidth]{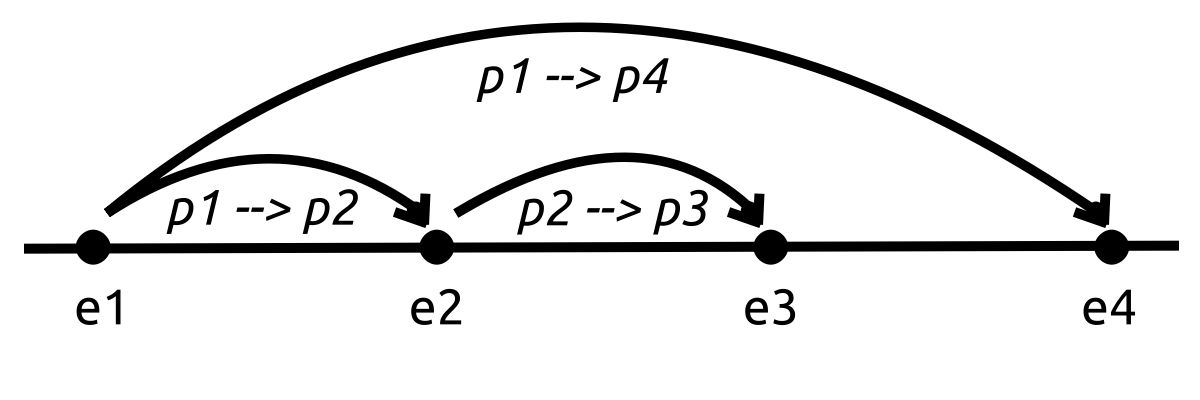}
    \caption{Example of a reliable broadcast protocol execution.}
    \label{fig:broadcast}
\end{figure}

Example~\ref{ex:broadcast} shows that the $\LCC$ metric diverges from the others in cases where a fast exchange of messages happens in the interval of one (or more) slow message.
This is the case for several $\ASO$ protocols in the literature (including ours) which heavily rely on relaying values to speed up the validation phase, making the metric unsuitable for our use case.
On the other hand, $\CR$ and $\NTR$ provide equivalent results in \emph{covered} executions, described next.\footnote{The Round and $\CR$ metrics also provide equivalent results in covered executions (Appendix~\ref{app:equivMetrics}).}

\subsection{Covered executions and holes}
\label{subsec:equiv}
Consider an execution $E = C_0e_1C_1...e_lC_l$ illustrated in Figure~\ref{fig:undefined} where no process receives a message from another process, i.e., events may add messages to the buffer but no event removes a message from it.
$\delta_{\overline{E}}$ is not defined in any time assignment $\overline{E}$.

Now consider an execution $E' = C_0e_1C_1...e_lC_l...e_mC_m$ in which:

\begin{itemize}
    \item A message $m$ is sent in $e_1$ and received in $e_l$;
    \item A message $m'$ is sent in $e_{l+1}$ and received in $e_m$;
    \item No message from $e_1...e_l$ is received in $e_{l+1}...e_m$.
\end{itemize}

In this example, illustrated in Figure~\ref{fig:unbounded} with $5$ events, $\delta_{\overline{E'}}$ exists for any time assignment of $E'$, but we can still assign an arbitrary time difference to $e_l$ and $e_{l+1}$ without affecting $\delta_{\overline{E'}}$, which results in the number of $\CR$ rounds to be unbounded.

\begin{figure}
    \centering
     \begin{subfigure}[t]{4.5cm}
         \centering
         \includegraphics[height=2.75cm]{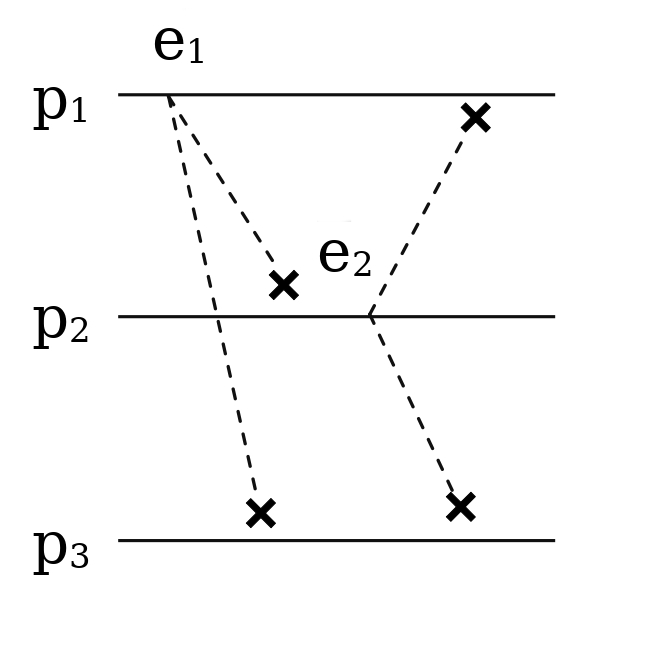}
         \caption{Execution with undefined $\delta_{\overline{E}}$.}
         \label{fig:undefined}
     \end{subfigure}
     \hfill
     \begin{subfigure}[t]{4.5cm}
         \centering
         \includegraphics[height=2.75cm]{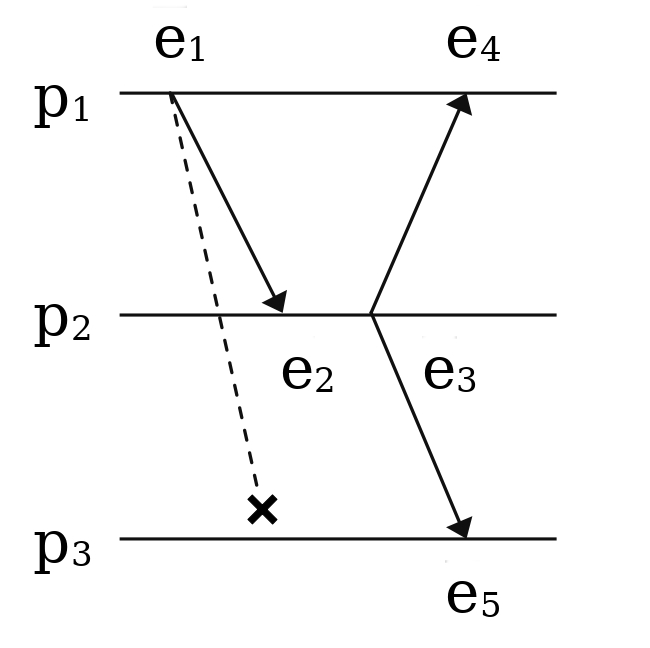}
         \caption{Execution where the number of rounds is unbounded according to Round and $\CR$.}
         \label{fig:unbounded}
     \end{subfigure}
     \hfill
     \begin{subfigure}[t]{4.5cm}
        \centering
        \includegraphics[height=2.75cm]{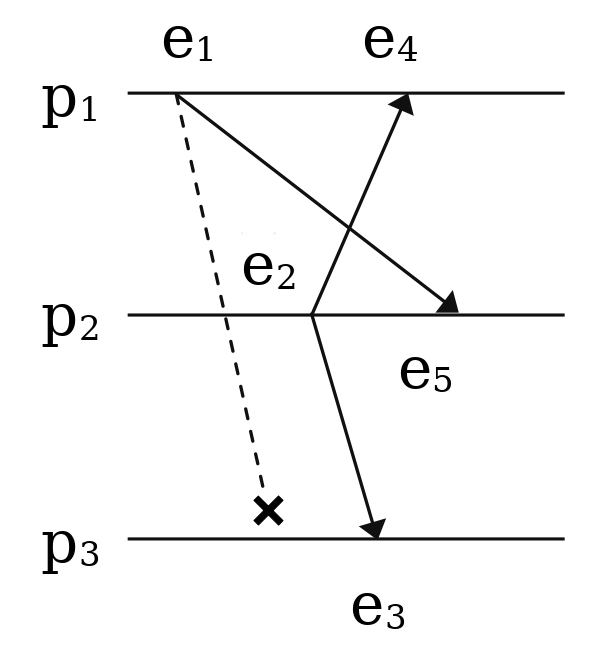}
        \caption{Covered execution.}
        \label{fig:covered}
     \end{subfigure}
    \caption{Examples of non-covered and covered executions.}
    \label{fig:executions}
\end{figure}

The two executions in the examples above have events whose time difference is unrelated to message delays.
By consequence, the duration of these executions can grow irrespective of any bound imposed by message exchanges.
Similarly, in Figure~\ref{fig:unbounded}, since there is no message being received in $e_3e_4e_5$ from $e_1e_2$, there is no round assignment defined when using $\NTR$ to $e_3$, $e_4$ and $e_4$.

We then restrict the analysis of these metrics to executions that are \emph{covered}. Formally:

\begin{definition} [Covered Execution]
\label{def:covExec}
    A \emph{hole} in an execution is a pair $(e_l,e_{l+1})$ in which no event in $e_{l+1}...$ receives a message from $...e_l$, in other words, there are no message hops among the two sequence of events. An execution is \emph{covered} iff it has no holes.
\end{definition}

Abraham et al.~\cite{abraham-metric} introduce $\NTR$ as an equivalent to $\CR$, however, no formal proof is provided.
The next result corroborates this claim in covered executions.
Later in Example~\ref{ex:diffCRNTR}, we show that using $\NTR$ naively to measure time between events may \emph{not} match $\CR$.

\begin{theorem}
    A finite covered execution $E$ has $k$ $\CR$ rounds iff it has $\lceil k \rceil$ $\NTR$ rounds.
\end{theorem}

\begin{proof}
    See Appendix~\ref{app:messageHops}.
\end{proof}

\subsection{Time between arbitrary events}

In long-lived executions (such as those of atomic snapshot algorithms) we are interested in measuring time between two events, for instance, between an application call and response.
Definition~\ref{def:crMetric} can easily be adapted to measure the number of rounds between two events as follows:

\begin{definition}[Generalized $\CR$ metric]
\label{def:genCR}
    Let $E$ be an execution, 
    let $\cT(E )$ denote the set of all timed executions $\overline{E}$ based on $E$, and $\delta_{\overline{E}}$ - the maximum message delay in $\overline{E}$.
    Let $e_i$ and $e_j$ ($j > i$) be events in $E$, and $t_i$ and $t_j$ time assignments in $\overline{E}$ for them respectively.
    Then we say that in between $e_i$ and $e_j$ there are:
    $\sup_{\overline{E} \in \cT( E )} (t_j-t_i)/ \delta_{\overline{E}}$
    \emph{$\CR$ rounds}.
\end{definition}

An appealing way of defining time between two events $e_i$ and $e_j$ using a non-timed metric is to assign rounds according to $\NTR$, and then take the difference of rounds assigned to $e_i$ and $e_j$.
As illustrated in Example~\ref{ex:diffCRNTR}, this definition can diverge from generalized $\CR$.

\begin{figure}[tp]
    \centering
    \includegraphics[width=0.5\textwidth]{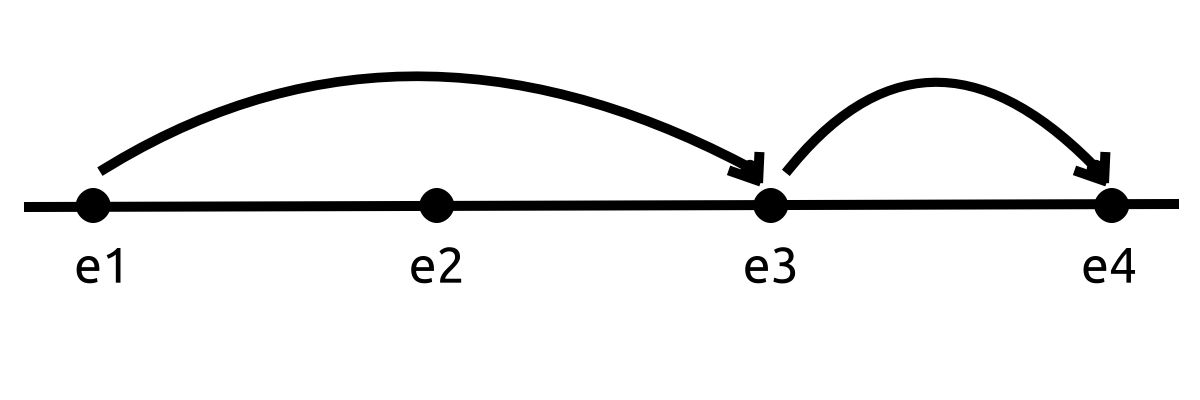}
    \caption{An execution in which there are $2$ $\CR$ rounds between $e_2$ and $e_4$. However, the difference of the rounds assigned to $e_2$ and $e_4$ using $\NTR$ is $1$.}
    \label{fig:eventCR}
\end{figure}

\begin{example}
\label{ex:diffCRNTR}
    Consider the execution shown in Figure~\ref{fig:eventCR}.
    We can assign times to $e_1$, $e_3$ and $e_4$ such that the two message hops have delay of $\delta$.
    Now consider the number of rounds between $e_2$ and $e_4$, since we can assign a time for $e_2$ that is arbitrarily close to $e_1$'s assignment, there are $2$ $\CR$ rounds between $e_2$ and $e_4$.
    However, the round assignments using $\NTR$ to $e_2$ and $e_4$ are $1$ and $2$ respectively, so simply taking the difference between them leads to a value that diverges from $\CR$.
\end{example}

We then give the following definition, using the approach described in Section~\ref{subsec:ourMetric}:

\begin{definition}[Generalized $\NTR$]
\label{def:genNTR}
Given an execution $E$, let $e_i$ and $e_j$ ($j > i$) be events in $E$. The number of rounds between $e_i$ and $e_j$ is given by the round assigned to $e_j$ according to the following:
\begin{itemize}
    \item All events up to (and including) $e_i$ are assigned round $0$. We also write $e_0^* = e_i$;
    \item For any $r \geq 1$, let $e_r^*$ be the last event where a message of round $r-1$ is delivered. All events after $e_{r-1}^*$ until (and including) event $e_r^*$ are in round $r$.
\end{itemize}
\end{definition}

\begin{theorem}
\label{th:evEquiv}
    Let $E$ be a covered execution and $e_i$ and $e_j$ ($j>i$) be events of $E$.
    There are $k$ rounds in between $e_i$ and $e_j$ according to $\CR$ (Definition~\ref{def:genCR}) iff there are $\lceil k \rceil$ rounds in between them according to $\NTR$ (Definition~\ref{def:genNTR}).
\end{theorem}

\begin{proof}
    See Appendix~\ref{app:arbEvents}.
\end{proof}

\subsection{Relating IRA to NTR}

\begin{theorem}
\label{th:IRAtoNTR}
    Let $E$ be a finite covered execution and suppose that all events of $E$ are assigned rounds according to $\IRA$ after all iterations of the algorithm.
    It holds that:

    \begin{enumerate}
        \item Round $0$ is composed only of $e_0$ (the initial event).
        \item The final event of round $i+1$ is the last event to receive a message from round $i$.
    \end{enumerate}
\end{theorem}

\begin{proof}
    See Appendix~\ref{app:thIRA}.
\end{proof}

\begin{corollary}
    \label{col:IRAisNTR}
    $\IRA$ and $\NTR$ assign the same rounds to events in covered executions.
\end{corollary}


\appendix

\section{Protocol with $O(n^2)$ message complexity per request}
\label{app:refinedAlg}

Algorithm~\ref{alg:longLA} has a complexity of $(n^2)$ messages per \emph{proposed} value, where each proposal might contain an arbitrary number of requests.
As a consequence, processes are required to exchange messages that can grow indefinetely in size, resulting in high network bandwidth usage.
We address this problem in Algorithm~\ref{alg:refinedlongLA}, with a few small modifications to Algorithm~\ref{alg:longLA}.

Instead of waiting for validation from a quorum and relaying entire proposals, processes keep track of each individual request, and relay only the difference between a received proposal and the current values waiting for validation.
The same occurs in the $\ACCEPT$ phase of the protocol, where processes only send the difference between new learned values and previous ones.
With these modifications, an individual request is now relayed once by every process before proposing, another time at the proposal and validation phase and one last time in the $\ACCEPT$ message, for a total of $3*n^2$ messages per request.

\begin{algorithm}
{\footnotesize
\begin{smartalgorithmic}[1]

\Upon{\Startup}
    \State $\Proposing, \MPool, \Pending, \Relaying, \Validated, \Learned, \ToAdopt \gets \emptyset$
\EndUpon
\algspace[0.5]

\Operation{\Propose}{$v$}
    \State \SendRequest($v$)
    \State wait until $v \sqsubseteq \bigsqcup\Learned$ 
    \State \Return $\bigsqcup\Learned$
\EndOperation
\algspace[0.5]

\Operation{\SendRequest}{$v$}
    \State $\MPool \gets \MPool \cup \{v\}$
    \State \Send{$\REQUEST, v$}{every other node}
\EndOperation
\algspace[0.5]

\UponReceiving{$\REQUEST, v$}{a node}
    \If{$v \not\in \MPool \cup \Proposing \cup \Learned$}
        \State $\MPool \gets \MPool \cup \{v\}$
        \State \Send{$\REQUEST, v$}{every other node}
    \EndIf
\EndUponReceiving
\algspace[0.5]

\UponEvent{$(\MPool \neq \emptyset) \wedge (\Proposing  = \emptyset)$}
    \State $\Proposing \gets \MPool$
    \For{$v \in \MPool$}
        \State $\Pending[v] \gets 1$
    \EndFor
    \State $\MPool \gets \emptyset$
    \State \Send{$\PREPARE, \Proposing$}{every other node}
\EndUponEvent
\algspace[0.5]

\UponReceiving{$\PREPARE, V$}{a node}
    \State $\Relaying \gets \emptyset$
    \For{$v \in V$}
        \If{$v \in \Pending.\keys()$}
            \State $\Pending[v]++$
        \Else
            \State $\Pending[v] \gets 1$
            \State $\Relaying \gets \Relaying \cup \{v\}$
        \EndIf
    \EndFor
    \If{$\Relaying \neq \emptyset$}
        \State \Send{$\PREPARE, \Relaying$}{every node}
    \EndIf
\EndUponReceiving
\algspace[0.5]

\Upon{exists $v$ s.t. $\Pending[v]=n-f$}
    \State $\Validated \gets \Validated \cup \{v\}$
\EndUpon
\algspace[0.5]

\UponEvent{$\Pending.\keys() \subseteq \Validated$}
    \If{$\Learned \subset \Validated$}
        \State $\Delta\Learned \gets \Validated - \Learned$
        \State $\ToAdopt \gets \ToAdopt - \Delta\Learned$
        \State $\Learned \gets \Validated$
        \State $\Proposing \gets \emptyset$
        \State \Send{$\ACCEPT,\Delta\Learned$}{every node}
    \EndIf
\EndUponEvent
\algspace[0.5]

\UponReceiving{$\ACCEPT, W$}{a node}
    \State $\ToAdopt \gets \ToAdopt \cup W$
    \If{$(\Proposing \subseteq \ToAdopt)$}
        \State $\Validated \gets \Validated \cup \ToAdopt$
        \State $\Learned \gets \Learned \cup \ToAdopt$
        \State $\Proposing \gets \emptyset$
        \State \Send{$\ACCEPT,\ToAdopt$}{every node}
    \EndIf
\EndUponReceiving
\end{smartalgorithmic}
}
\caption{Refined Long-Lived LA: code for node $x$.}
\label{alg:refinedlongLA}
\end{algorithm}

\section{Time Complexity of Algorithm~\ref{alg:longLA}}
\label{app:ourLAcomp}

We show that an operation $\Op$ in Algorithm~\ref{alg:longLA} takes $O(k)$ rounds to complete, where $k$ is the number of active faulty nodes during $\Op$.

Messages from and to faulty nodes may not arrive, however, a message sent by (and to) a faulty node at round $r$ is received at most by round $r+1$.
Moreover, since channels are FIFO, when a node $i$ receives a message from another node $j$, $i$ must also have received all previous messages $j$ sent to $i$, irrespective of them being correct or faulty.

If a correct node receives $\langle \PREPARE, v' \rangle$ (even from a faulty node) in round $r$, every correct node will have $v'$ added to $\Pending$ by the end of round $r+1$, and will have $v'$ validated by the end of round $r+2$.
Also, faulty nodes wait for its current proposal to finish before starting a new one, in which case they send an $\ACCEPT$ message for the last learned value before sending the new proposal.

We say that a node \emph{introduces a new value} $w$ during the operation if it is the first node to send a $\langle \PREPARE, w \rangle$ for $w$ in the interval of the operation.
A node can introduce a new value coming from an internal source, i.e., the value was buffered and proposed when the node had already finished its previous proposal, or from an external source, i.e., after receiving a proposal originated from another node before the operation started.

Let $v$ be the value received from the application call for $\Op$ and $e_C$ (as well as all previous events) be assigned round $0$.
If there are no active faulty nodes, a correct node learns a value containing $v$ by at most round $7$ (by Lemma~\ref{lm:timeLearn}, here, we include the time $v$ can remain buffered).
Also by the end of round $5$, every correct node has sent a $\PREPARE$ message for $v$ and has $v$ validated by the end of round $6$ (including buffering time, a correct node proposes $v$ in round $4$ at the latest).
By that point, all correct nodes are waiting for their proposals to complete and, therefore, cannot introduce a value from an internal source.
In order to delay a correct node from leaning a value containing $v$ by round $7$, every correct node should receive a new value in a $\PREPARE$ message before, which is added to $\Pending$ but is not validated.
Using a simple inductive argument, $2k+1$ new proposals originated from faulty nodes are necessary to delay a correct node from learning a value from round $7$ to round $7 + 2k$.

Suppose that there is an execution where it takes $8 + 2k + 1$ rounds for node $i$ to complete an operation.
But there are only $k$ active faulty nodes, which means that at least $k+1$ extra proposals were introduced by active faulty nodes.

Let $f_0$ be an active faulty node that introduced more than one of the $2k+1$ values that delayed the operation (assuming w.l.o.g. that there are \emph{exactly} $2k+1$ new proposals).
Let $w$ and $w'$ be the first and the second values introduced by $f_0$ respectively.
If $w'$ was received by $f_0$ from an internal source, $f_0$ should have finished its previous proposal (and learned a value containing $w$) before proposing $w'$.
But because $w$ was one of the values that delayed the operation, and since channels are FIFO, $f_0$ needs to add $v$ to $\Pending$ before validating $w$ (at least a majority of correct nodes sent a $\PREPARE$ for $v$ before sending a $\PREPARE$ for $w$).
$f_0$ then learns a value containing $v$ and sends $\ACCEPT$ with that value to everyone.
The $\ACCEPT$ message is received by correct processes before $\langle \PREPARE, w' \rangle$, and they would be able to adopt it.

So $f_0$ must have received $\langle \PREPARE, w' \rangle$ from an external source at most by round $1$, which means it issued proposals for $w'$ that can be received by at most round $2$.
We can also conclude that at least $k+1$ values were introduced by active faulty nodes from external sources.
Now let $w_{k+1}$ be the $(k+1)$th such value used to delay correct nodes from learning $v$.
The earliest round $w_{k+1}$ can delay is $7+k$, which means that by round $7+k$ all correct nodes already sent a propose for $w_{k+1}$, but by the end of round $5+k$ no correct node has done it (otherwise $w_{k+1}$ would have been validated in round $7+k$ by every correct process).
Take the first active faulty node $f_1$ from which a correct node received $\langle \PREPARE, w_{k+1} \rangle$.
Since the earliest this message is received is in round $6+k$, the earliest it could be sent is in round $5+k$, so $f_1$ first received $\langle \PREPARE, w_{k+1} \rangle$ from another distinct active faulty node, $f_2$, which sent it in round $4+k$ the earliest.
But $w_{k+1}$ was introduced from an external source and it needs to be received by a faulty node at round $1$.
Following the chain above, for the node $f_{k+6}$ to receive it in round $1$, there would be necessary a chain of $k+6$ active nodes, although there are only $k$.

Therefore, an operation takes less than $8+2k+1$ rounds to complete.

\section{Equivalence Proofs for Time Measurement Metrics}
\label{app:equivMetrics}
In this section, we present detailed proofs for the equivalence between $\CR$ and $\NTR$ in covered executions.
The proofs are written with respect to a new (non-timed) method for interpreting latency: the minimum number of \emph{message hops} that can cover an execution.
Before proceeding, we establish the equivalence between the Round and $\CR$ metrics.

\begin{theorem}
    Round and $\CR$ assign the same number of rounds to finite covered executions.
\end{theorem}

\begin{proof}
    Let $E$ be a finite covered execution and $\overline{E}$ a time assignment for $E$, with maximum message delay $\delta_{\overline{E}}$.
    Since we consider algorithms that do not make use of clocks, we can ``shrink'' or ``stretch'' time assignments without altering the steps in the underlying execution.
    Consider the time assignment $\overline{E}'$ built as following:

    \begin{enumerate}
        \item $t_{\Start}(\overline{E}') = t_{\Start}(\overline{E})$;
        \item For every event $\overline{e}_{l}$ in $\overline{E}$ with time $t_{l}$, have $\overline{e}_{l}'$ in $\overline{E}'$ with time $t_l'=t_{\Start} + (t_{l} - t_{\Start})\frac{1}{\delta_{\overline{E}}}$.
    \end{enumerate}

    We call $\overline{E}'$ a \emph{normalization} of $\overline{E}$.
    By construction, the maximum message delay in $\overline{E}'$ is $1$ and $\overline{E}'$ has the same number of $\CR$ rounds than $\overline{E}$.

    Now let $\cT(E)$ be the set of valid executions for the Round metric and $\overline{E} \in \cT(E)$ have $k$ rounds (using Round).
    If $\delta_{\overline{E}} < 1$, then the normalization $\overline{E}'$ of $\overline{E}$ has $k'>k$ rounds: $\Duration(\overline{E}') = t_{end}(\overline{E}')-t_{start}(\overline{E}') = (t_{end}(\overline{E})-t_{start}(\overline{E}))/\delta_{\overline{E}}$.

    Consider the set $\cT'(E)$ of valid executions for the Round metric where for all $\overline{E}' \in \cT'(E)$, $\delta_{\overline{E}'} = 1$.
    Since for every timed execution $\overline{E} \in \cT(E)$ with $k$ rounds, there is a timed execution $\overline{E}' \in \cT'(E)$ with $k'$ rounds where $k' \geq k$, then: $\sup_{\overline{E}' \in \cT'(E)}{\Duration(\overline{E}')} = \sup_{\overline{E} \in \cT(E)}{\Duration(\overline{E})}$.

    Easy to see that every execution $\overline{E}' \in \cT'(E)$ has the same number of rounds according to both $\CR$ and Round metrics.
    So if the Round metric assigns $k$ rounds to $E$ and $\CR$ assigns $k'$, $k' \geq k$.
    But we also know that for any time assignment $\overline{E}$ of $E$, the normalization of $\overline{E}$ is a valid timed-execution for the Round metric and has the same number of $\CR$ rounds as $\overline{E}$.
    This means that $k \geq k'$, since for any time assignment $\overline{E}$, there is a time assignment with the same number of rounds in both Round and $\CR$ metrics,
    therefore $k = k'$.
\end{proof}

\subsection{A new look at execution latency}
\label{app:messageHops}

In covered executions, it seems natural to relate the number of message hops to the number of communication rounds.
Next, we define the concept of covering executions and events with message hops.

Consider the finite execution $E = e_0,\ldots,e_l$.
We can visualize these events as points on a real line, where their positions correspond to their indices, that is, $e_0$ at $0$, $e_1$ at $1$ and so on.
Each pair of events $(e_i,e_j)$ defines an interval $[i,j]$, and we denote this by $\interval((e_i,e_j)) = [i,j]$.
Likewise, $E$ defines the interval $[0,l]$, which we represent as $\interval(E) = [0,l]$.

Since a message hop consists of a pair of events $(e_i,e_j)$, it also specifies an interval $[i,j]$.
For a set $M$ of message hops, we define $\interval(M) = \bigcup_{m \in M}\interval(m)$.

\begin{definition}[Execution cover]
\label{def:seqCover}
Let $E$ be a finite execution and $M$ a set of message hops from $E$.
We say that $M$ \emph{covers} $E$ if $\interval(M) = \interval(E)$.
Analogously, we say that $E$ can be covered by $k$ message hops if $|M| = k$.
\end{definition}

\begin{theorem}
\label{th:messageCover}
    If a covered execution $E$ has $k$ rounds according to the Round metric, then $\lceil k \rceil$ message hops are necessary and sufficient to cover $E$.
\end{theorem}

\begin{proof}
     Let $E$ have $k$ rounds according to Round.
    There is a time assignment $\overline{E}$ with $\Duration(\overline{E}) = k - \epsilon$, where $\epsilon > 0$ can be arbitrarily small.
    Starting from $t_{\Start}$, assume that there exists a set of message hops where each hop can cover the maximum amount of time, that is, an interval of one unit, then at least $\lceil \Duration(\overline{E}) \rceil = \lceil k \rceil$ message hops are necessary to cover the whole duration.

    Now, we proceed to build a set $M$ that covers $E$ with $\lceil k' \rceil$ message hops, and show next that there exists a timed execution with $k''$ rounds, where $\lceil k'' \rceil = k'$.

    For the first element of $M$, we take pair $p_1 = (e_1',e_1^*)$ where $e_1'=e_0$ (the initial event) and $e_1^*$ is the last event in $E$ where a message from $e_1'$ is received, we also define $e_0^* = e_0$.
    Now, we inductively take pair $p_i = (e_i',e_i^*)$ where $e_i^*$ is the last event to receive a message originated in $e_{i-2}^*...e_{i-1}^*$ and $e_i'$ is the \emph{first} corresponding event to have sent such message ($e_i^*$ may receive more than one).
    We continue to select pairs until pair $p_{k'} = (e_{k'}',e_{k'}^*)$ where $e_{k'}^*$ is the last event of $E$ (note that this construction is possible since $E$ is covered).

    The set $M = \{(e_1',e_1^*),...,(e_{k'}',e_{k'}^*)\}$ clearly covers $E$, implying that $E$ can be covered by $k'$ message hops.
    We now show that there exists a time assignment $\overline{E}$ with $k''$ rounds in which $\lceil k'' \rceil = k'$.

    Consider the following time assignment:

    \begin{itemize}
        \item $t_0 = t_1' = t_0^* = 0$
             \item $t_1^* = 1$
    \end{itemize}

    Take the sub-sequence $E_1$ containing all events in $e_0^*...e_1^*$ except for $e_0^*$.
    Note that $e_2'$ appears in $E_1$ and by construction, every message originated in $E_1$ that is received in the execution is received before $e_2^*$ or at $e_2^*$.

    We now enumerate the events in $E_1$ in reverse order: $e_1^*$ is assigned $0$, the event preceding $e_1^*$ receives $1$ and so on until the first event of $E_1$ receives $n_1$.
    Assign time to these events according to their enumeration $j_1$ as following:

    \begin{itemize}
        \item $t_{j_1} = t_1^* - \epsilon_1j_1$, where $0 < \epsilon_1n_1 < t_1^*$ if $n_1 > 0$ and $\epsilon_1 = 0$ otherwise.
    \end{itemize}

    We make so that $t_2^* = t_1^* - \epsilon_1n_1 + 1$, so that every message hop originated from $E_1$ satisfies the upper bound on message delay.

    In general, let $E_i$ be the sub-sequence containing all events in $e_{i-1}^*...e_i^*$ except for $e_{i-1}^*$.
    Enumerate the events in $E_i$ in the following order: $e_{i-1}^*$ receives $0$, the event preceding it receives $1$ and so on until the first event in $w_i$ receives $n_i$.
    Assign time to these events according to their enumeration $j_i$ as follows:
    \begin{itemize}
        \item $t_{j_i} = t_i^* - \epsilon_ij_i$, where $0 < \epsilon_in_i < (t_{i}^* - t_{i-1}^*)$ if $n_i >0$ and $\epsilon_i = 0$ otherwise.
    \end{itemize}

    We make so that $t_i^* = t_{i-1}^* - \epsilon_{i-1}n_{i-1} + 1$.
    For simplicity, assume that every $n_i > 0$ (in the case where some $n_i = 0$, we just make $t_i* = t_{i-1}^* + 1$ and the following analysis works analogously).
    From the time assignments above:
    \begin{gather}
        \label{eq:duration}
        \Duration(\overline{E}) = k'' = t_{k'}^* - t_0^* \\
        \label{eq:tEnd}
        t_{k'}^* = k' - (\epsilon_1n_1 + ... + \epsilon_{k'-1}n_{k'-1})
    \end{gather}

    With the following constraint, for all $i = 1,...,k'$ (we make $\epsilon_0n_0 = 0$):
    \begin{equation}
        \label{eq:constraint}
        0 < \epsilon_in_i < 1 - \epsilon_{i-1}n_{i-1}
    \end{equation}

    Let us make $\epsilon = \epsilon_1 = \epsilon_2 = ... = \epsilon_{k'-1}$, and let $n_{max} = max(n_1,...,n_{k'-1})$.
    The conditions in (\ref{eq:constraint}) can be satisfied by making:
    \begin{gather}
        \label{eq:maxCond}
        \epsilon n_{max} < 1 - \epsilon n_{max} \\
        \label{eq:firstCond}
        \epsilon < \frac{1}{2n_{max}}
    \end{gather}

    In order to make $\lceil k'' \rceil = k'$, the difference between $k'$ and $k''$ needs to be in the interval:
    \begin{equation}
        0 \leq (k' - k'') < 1
    \end{equation}
         
    Thus,
    \begin{gather}
        k' - k'' = k' - k' + (\epsilon_1 n_1 + \ldots + \epsilon_{k'-1} n_{k'-1}) < 1 \\
        \label{eq:mainCond}
        (\epsilon n_1 + \ldots + \epsilon n_{k'-1}) < 1 \\
        (\epsilon n_1 + \ldots + \epsilon n_{k'-1}) < (k'-1)\epsilon n_{max}
    \end{gather}

    To satisfy (\ref{eq:mainCond}), we can make so that:
    \begin{gather}
        (k'-1)\epsilon n_{max} < 1 \\
        \label{eq:secondCond}
        \epsilon < \frac{1}{(k'-1)n_{max}}
    \end{gather}
    
    From (\ref{eq:firstCond}) and (\ref{eq:secondCond}):
    \begin{equation}
        \label{eq:voila}
        \epsilon < min(\frac{1}{2n_{max}},\frac{1}{(k'-1)n_{max}})
    \end{equation}

    As long as inequality (\ref{eq:voila}) is satisfied, the time assignments we have chosen guarantee that $\lceil k'' \rceil = k'$.
    Because $k'' \leq k$ and $k'$ message hops cover $E$ for any time assignment, it also follows that $\lceil k \rceil = k'$.
\end{proof}

\begin{corollary}
\label{col:CRtoHops}
    If a covered execution $E$ has $k$ $\CR$ rounds, then $\lceil k \rceil$ message hops are necessary and sufficient to cover $E$.
\end{corollary}

The next results corroborate the equivalence between $\NTR$ and $\CR$.

\begin{theorem}
\label{th:ntrCover}
    Let $E$ be a finite covered execution.
    $E$ has $k$ rounds in the $\NTR$ metric \emph{iff} $k$ message hops are necessary and sufficient to cover it.

    \begin{proof}
        Let events in $E$ be assigned rounds according to $\NTR$, resulting in $k$ rounds. We can select a set $M$ of $k$ message hops as following (sufficiency):
        \begin{itemize}
            \item Take the first pair $p_1 = (e_0^*,e_1^*)$;
            \item Take pair $p_i = (e_i',e_i^*)$, where $e_i^*$ is the last event of round $i$, and $e_i'$ is the first event from which a message is received in $e_i^*$ ($e_i'$ has to be an event of round $i-1$).
        \end{itemize}
        Since there are $k$ rounds, $M$ has $k$ message hops.
        It is also easy to see that $M$ covers $E$.

        Suppose that a sequence $M$ of $k'$ message hops can cover $E$ with $k'<k$.
        If we assume that each pair $(e_l,e_m)$ are assigned either with the same number of rounds or $e_m$ has one round higher than $e_l$, then since $k' < k$, there would be an entire round that is not covered by any message hop.
        On the other hand, a pair $(e_l',e_m')$ cannot have $e_m'$ assigned two (or more) rounds higher than $e_l'$ by definition, since $e_m'$ receives a message from $e_l'$ (necessity).

        Now let $k$ message hops be necessary and sufficient to cover $E$, and assume that $E$ has $k'$ $\NTR$ rounds.
        Then $k = k'$, since $k'$ rounds are necessary and sufficient to cover $E$.
    \end{proof}
\end{theorem}

\begin{corollary}
    A finite covered execution has $k$ $\CR$ rounds iff it has $\lceil k \rceil$ $\NTR$ rounds.
\end{corollary}

\subsection{Latency between arbitrary events}
\label{app:arbEvents}
We generalize Definition~\ref{def:seqCover} to account for the time between any two events in an execution.

\begin{definition}[Event cover]
Let $E$ be a finite execution, $e_i$ and $e_j$ ($j>i$) events in $E$ and $M$ a set of message hops from $E$.
We say that $M$ \emph{covers} $(e_i,e_j)$ if $\interval(e_i \ldots e_j) \subseteq \interval(M)$.
Analogously, we say that $(e_i,e_j)$ can be covered by $k$ message hops if $|M| = k$.
\end{definition}

\begin{theorem}
\label{th:genCRtoHops}
    Let $E$ be a covered execution and $e_i$ and $e_j$ be events in $E$. There are $k$ $\CR$ rounds in between $e_i$ and $e_j$ iff $\lceil k \rceil$ message hops are necessary and sufficient to cover them.
\end{theorem}

The proof of Theorem~\ref{th:genCRtoHops} is similar to that of Theorem~\ref{th:messageCover}
and is omitted (we can consider a covered sub-sequence of $E$ with $k$ rounds as a covered execution).

\begin{theorem}
    \label{th:eqNTRevent}
    Let $E$ be a covered execution and $e$ and $e'$ be events in $E$. If there are $k$ rounds in between $e$ and $e'$ according to $\NTR$ (Definition~\ref{def:genNTR}) then $k$ message hops are necessary and sufficient to cover $e$ and $e'$.
\end{theorem}

\begin{proof}
    Let $e$ be assigned round $0$ (as well as all previous events) and $e'$ round $k$.
    Take $e_1^*$, the last event of round $1$, and the earliest event $e_1'$ from which $e_1^*$ received a message.
    Since $e_1^*$ receives a message from round $0$, $e_1'$ must be assigned round $0$.

    Inductively, take $e_i^*$, the last event of round $i$, and the earliest event $e_i'$ from which $e_i^*$ receives a message.
    Since $e_i^*$ receives a message from round $i-1$ (by definition), $e_i'$ must be assigned round $i-1$.

    Consider the set $M = \{(e_1',e_1^*),\ldots,(e_k',e_k^*)\}$, $M$ clearly covers $e$ and $e'$ (sufficiency).

    Now consider a set $M'$ with $k'$ message hops such that $M'$ covers $e$ and $e'$.
    Since $M'$ covers the two events, there must be a message hop whose first event (the sender event) is in round $0$.
    This is true for any round up to $k-1$: suppose that there is a round $i$ where no message hop in $M'$ has the first event in round $i$, then since $e,\ldots,e'$ is covered, there exists a message originated from a previous round $j<i$ that is received in a round $l>i$. But then $l \leq i + 1$ by definition of the metric, a contradiction.
    Thus, $M'$ includes at least one message hop for each round from $0$ to $k-1$, so $k' \geq k$ (necessity).
\end{proof}

\begin{corollary}
    Let $E$ be a covered execution and $e_i$ and $e_j$ be events in $E$.
    There are $k$ $\CR$ rounds in between $e_i$ and $e_j$ iff there are $\lceil k \rceil$ rounds in between them according to $\NTR$.
\end{corollary}

Finally, we prove Theorem~\ref{th:IRAtoNTR}, relating $\IRA$ to $\NTR$.

\subsection{Proof of Theorem~\ref{th:IRAtoNTR}}
\label{app:thIRA}

Let $E$ be a finite covered execution and suppose that all events of $E$ are assigned rounds according to $\IRA$ after all iterations of the algorithm.
It holds that:

\begin{enumerate}
    \item Round $0$ is composed only of $e_0$ (the initial event).
    \item The final event of round $i+1$ is the last event to receive a message from round $i$.
\end{enumerate}

\begin{proof}
    1. The case where $E$ has a single event is immediate, next, we consider executions with more than one event.
    From the algorithm, $e_0^*=e_0$ ($e_0^*$ does not change).
    Since $E$ is covered, there is at least one event which receives a message from $e_0$.
    Let $e'$ be the last such event.
    When the algorithm arrives at the iteration for $e'$, since the oldest message is from round $0$ (from $e_0$), all events after $e_0^*$ until $e'$ are assigned round $1$ (line~\ref{line:defRound}).
    Since no event can receive round $0$ in later iterations, $e_0$ is the only event remaining with round $0$ assigned.

    2. As shown above, there is a single event in round $0$.
    Let $e'$ be the last event to receive a message from $e_0$, in its iteration $e'$ then receives round $1$.
    The events following $e'$ (assuming $e'$ is not the last event) might momentarily be assigned to round $1$ (if they do not receive any message, line~\ref{line:momentRound}), but since the execution is covered, there must be an event after $e'$ that receives a message from $e_1...e'$.
    Let $e''$ be the last such event, in its iteration, $e''$ is assigned round $2$, and all events after $e'$ (which is the last event $e_1^*$ is assigned to, line~\ref{line:lastStar}) also receive round $2$.
    No later iteration can assign round $1$ to those events since no other event receives a message from round $0$, thus $e'$ is the last event to remain with round $1$ assigned.

    Now assume that the final event $e_i^*$ of round $i$ is the last to receive a message from round $i-1$, and that there is an event assigned to round $i+1$.
    Suppose that $e^*$, the last event to receive a message from round $i$, is not the final event of round $i+1$.
    Since $e^*$ receives a message from round $i$ but not from an event before round $i$, it has to be assigned round $i+1$ and $e_{i+1}^*$ receives $e^*$ in line~\ref{line:lastStar} of the algorithm.
    It follows that the final event of round $i+1$ comes after $e^*$ and receives no message from $e_0,\ldots,e_i^*$.
    Because the execution is covered, there must be at least one event after the final event of round $i+1$ that receives a message from round $i+1$.
    Once more, consider the last such event, so all events after $e_{i+1}^*$ until this event are assigned round $i+2$, leaving $e^*$ as the final event of round $i+1$.
\end{proof}

\section{One-Shot Lattice Agreement}
\label{app:oneShot}
In the \emph{One-Shot Lattice Agreement} problem, every process starts with the proposal of an initial value and terminate when it learns a value, such that $\GValidity$, $\GConsistency$ and $\GLiveness$ are satisfied (Section~\ref{sec:definitions}).
In this section, we analyze time complexity of one-shot $\LA$ protocols, as the abstraction can be used as a building block for implementing $\ASO$~\cite{la95,garg2020amortized}.

In every protocol execution, all processes start proposing a value simultaneously, i.e., in the initial event.
We measure the time for all correct processes to learn a value in the fault-free and worst-case latency.
In fault-free executions, all processes are correct and every message sent in the execution must arrive.
On the other hand, in the worst-case, there is a set of correct processes $P$ and a set of potentially faulty processes $F$, where $P$ has $f+1$ processes and $F$ has $f$ processes.
All messages exchanged within $P$ arrive, but this is not the case for exchanges within $F$ or between $P$ and $F$.

We show that: 1) the protocol presented in~\cite{faleiro2012podc} has a constant latency in fault-free runs, as opposed to the $O(n)$ complexity claimed in the paper.
2) With the conventional model of reliable channels assumed in this paper, the protocol of~\cite{garg2020amortized} has $\Omega(f)$ time complexity in the worst-case, as opposed to $O(\sqrt{f})$ when assuming their model.
3)~\cite{imbs2018set}'s protocol has $\Omega(f)$ time complexity in the worst-case, which is not analyzed in their paper.

\subsection{One-Shot Lattice Agreement by Faleiro et al.~\cite{faleiro2012podc}}

Figure~\ref{fig:faleiro-osla} shows the one-shot $\LA$ description which was extracted from~\cite{faleiro2012podc}.
The protocol describes the roles of proposers and acceptors, but we assume that all processes perform both roles.

A proposer proceeds in rounds, were each round consists in sending a proposed value to every acceptor and waiting for the reply from a majority of them.
If all replies are acknowledgments, the process can learn the current proposed value.
On the other hand, if there is a \textit{NACK} with an unseen value, the proposer joins it with the previously proposed value and re-sends them.

An acceptor stores the join of every proposed value it receives.
When a proposal is received such that it contains all the stored values, the acceptor replies with an acknowledgment, otherwise it sends the stored values back to the proposer in a \textit{NACK} message.

\begin{figure}[!htp]
    \centering
    \includegraphics[width=\textwidth, page=4, trim = {5em 17em 13em 5em}, clip]{related_work/generalized-podc2012.pdf}
    \caption{One-shot LA algorithm as presented in~\cite{faleiro2012podc}.}
    \label{fig:faleiro-osla}
\end{figure}

\begin{theorem}
\label{th:faleiroOneShot}
    The one-shot $\LA$ protocol of~\cite{faleiro2012podc} takes at most $6$ rounds in fault-free runs.

    \begin{proof}
        All processes start their proposal in the initial event, which is in round 0.
        Regardless of the order of messages, in the last event of round 1, every process will have received everyone's first proposal and their local $\acceptedValue$ is a join of all initial values.
        So every reply made in round $2$ onward will contain all values.

        Consider a process $p$, every process receives $p$'s proposal in round $1$ and reply, so that $p$ refines its proposal (and propose again) in round $2$ at most.
        If $p$ re-proposes in round $1$, suppose that $Q$ is the set of processes from which $p$ receives the replies (for this refined proposal), then either: no reply from $Q$ is made in round $2$ (only round $1$), or some reply is from round $2$.

        In the first case, since all replies come from round $1$, the refinement and new proposal must happen in either round $1$ (in which case we come back to the situation above) or round $2$.
        In the second, $p$ receives all values and re-propose in at most round $3$,
        and since the proposal contains all values, $p$ learns a value by at most round $5$.

        Now the only remaining case is when $p$ initiates a new proposal in round $2$ with some value missing.
        In this case all replies will be a join of all values, and by at most round $4$, $p$ will refine the proposal, learning a value by at most round $6$.
    \end{proof}
\end{theorem}

    \subsection{One-Shot Lattice Agreement by Garg et al.~\cite{garg2020amortized,ASOconstant}}
    \label{app:gargOneShot}

    Garg et al.~\cite{garg2020amortized,ASOconstant} assume a stronger underlying reliable channel for communication than in this paper.
    In their papers, the channel is responsible for delivering a message sent from one process to another, thus messages sent by faulty processes (to correct ones) are guaranteed to arrive in an infinite execution.
    In the following, we analyze their protocol under the more conventional assumption that messages from faulty processes may never be received.

    Figure~\ref{fig:garg-la} (extracted from~\cite{garg2020amortized}) shows a description of their one-shot $\LA$.
    Every process $i$ maintains a local \emph{view array}, where each position $j$ in the array contains the values $i$ received from $j$.
    In the start of the protocol, every process sends its initial value to everyone.
    Processes relay (execute block from lines $5$ to $7$) every new value they receive from other processes, and can learn a value once their local view satisfy a predicate called \emph{equivalence quorum}.
    Intuitively, the local view $V$ of process $i$ satisfies the predicate if there is a quorum in which $V[i] = V[j]$ for every process $j$ in the quorum.
    
    \begin{figure}[!htp]
        \centering
        \includegraphics[width=\textwidth, page=4, trim = {10.5em 35em 9em 27em}, clip]{related_work/2008.11837v2.pdf}
        \caption{LA algorithm as presented in~\cite{garg2020amortized}.}
        \label{fig:garg-la}
    \end{figure}

    \begin{theorem}
    \label{th:goodGargLA}
        The one-shot $\LA$ protocol of~\cite{garg2020amortized} takes at most $2$ rounds in fault-free runs.

        \begin{proof}
            Every process sends their initial value in round $0$.
            By the end of round $1$, every process has already received and relayed all other values, so that by the end of round $2$, all the local views contains every value, resulting in all processes learning a value.
        \end{proof}
    \end{theorem}

    \begin{theorem}
        The one-shot $\LA$ protocol of~\cite{garg2020amortized} has $\Omega(f)$ worst-case latency.

        \begin{proof}
            We proceed to build an execution that takes at least $f/2$ rounds to complete.
            Assume w.l.o.g. that the number of faulty processes is even.
            Split $F$ into two groups $A = \{l_{1},\ldots,l_{f/2}\}$ and $B = \{l_{f/2 + 1},\ldots,l_{f}\}$ with $f/2$ processes each. In round $0$, every process sends its initial value to everyone.

            [Round $1$] At the beginning of the round, the value $(x_1,l_1)$ from $l_1$ is received and relayed by every process in $B$, as well as by a single correct process $l_c$.
            All remaining values $(x_i,l_i)$ from processes in $A$ are received by a single process $l_{f/2+1} \in B$, which relays them.
            Processes in $A$ crash just after $l_{f/2+1}$ receives their values, and no other process receives any message from them.
            At the end of the round, initial values from every non-crashed process (including those in $B$) are received and relayed by every non-crashed process.

            [Round $2$] At the beginning, the first of the remaining values $(x_2,l_2)$ that $l_{f/2+1}$ relayed is received (and relayed) by every non-crashed process in $B$, as well as by $l_c$.
            A single process $l_{f/2+2}$ receives all other $f/2 - 2$ values from $l_{f/2+1}$ and relay them,
            then $l_{f/2+1}$ crashes and no other process receives messages from it.
            Finally, every non-crashed process receives the values relayed by other non-crashed processes in the previous round.
            By the end of round $2$, any non-crashed process have in its view $V[j]$ all initial values sent by non-crashed processes, but for a single correct process $l_c$ and all non-crashed processes in $B$, their position in the view also contains $(x_1,l_1)$.
            Therefore, no equivalence quorum exists in any local view.

            [Round $i+1$] At the beginning, the first of the remaining values $(x_{i+1},l_{i+1})$ that $l_{f/2+i}$ relayed is received (and relayed) by every non-crashed process in $B$, as well as by $l_c$.
            A single process $l_{f/2+i+1}$ receives all other $f/2-(i+1)$ values from $l_{f/2+i}$ and relay them, then $l_{f/2+i}$ crashes and no other process receives messages from it.
            At the end of the round, every non-crashed process receives the values relayed by other non-crashed processes in the previous round, but messages from $l_c$ are received before any other.
            Any non-crashed process have in its view $V[j]$ all initial values sent by non-crashed processes, with addition of $(x_1,l_1),\ldots,(x_{i-1},l_{i-1})$, but for a single correct process $l_c$ and all remaining non-crashed processes in $B$, their position in the view also contains $(x_{i},l_{i})$.
            Note that, since messages from $l_c$ are received first, non-crashed processes receive and relay $(x_i,l_i)$ before forming an equivalence quorum for previous values, and no process is able to learn a value this round.

            We can use the above method to delay the execution by $f/2$ rounds.
        \end{proof}
    \end{theorem}

    \subsection{One-Shot Lattice Agreement by Imbs et al.~\cite{imbs2018set}}

    The protocol displayed in Figure~\ref{fig:imbs-osla} (extracted from~\cite{imbs2018set}) solves the problem of Set-Constrained Delivery Broadcast (SCD-Broadcast). It can easily be adapted to solve one-shot $\LA$ by adding to the condition of line $17$ that the initial value must be in the output.

    The authors use a FIFO broadcast primitive for forwarding messages, so in the following proofs we will assume message channels to be FIFO.
    We say that a process \emph{relays} a value when it executes line $11$ of the algorithm (it sends a new received value to everyone).
    The fundamental blocks of the protocol include:

    \begin{itemize}
        \item Each process has a logical clock which ticks every time a new value is received and relayed (including its own initial value).
        The current clock value is attached to the relaying message (called a \emph{forward} message).

        \item Each process stores a set of value views: an array of logical clock values, one position for each process.

        \item The following predicate needs to hold in order to output a set o values $O$: Let $A$ be the set of all received values and $V$ be the set of values received by a quorum.
        An output is a non-empty set $O \subseteq V$ satisfying: $\forall w \in O, \forall v \in A-O:$ there is a quorum in which each individual clock value for $w$ is smaller than the corresponding value for $v$.
    \end{itemize}

    \begin{figure}[!htp]
        \centering
        \includegraphics[page=6, trim = {13em 44.5em 13em 8em}, clip]{related_work/SCD.pdf}
        
        \caption{SCD algorithm as presented in~\cite{imbs2018set}.}
        \label{fig:imbs-osla}
    \end{figure}

    Not that each process starts sending its initial value to itself before relaying it to everyone.
    For simplicity, we consider these two actions to be in a single event (the initial event), where the first message sent to itself is ignored.
    
    \begin{theorem}
        The SCD-Broadcast protocol in~\cite{imbs2018set} takes at most $2$ rounds in the fault-free runs

        \begin{proof}
            In the first event (round $0$), every process forwards its own initial value to everyone.
            At the end of round $1$, all processes have already received all initial values and relayed them.
            At the end of round $2$, regardless of the order, all processes received all values from everyone.
            As a consequence $A-V = \emptyset$ in their local view, so all processes can output $V$.
        \end{proof}
    \end{theorem}

    \begin{theorem}
        The SCD-Broadcast protocol in~\cite{imbs2018set} has $\Omega(f)$ worst-case latency.

        \begin{proof}
            We proceed to build an execution that takes at least $f/4$ rounds to complete.
            When we say that a process \emph{crashes} at some point in the execution, the process no longer takes any more steps and no further messages are received from it unless explicitly stated.

            Assume w.l.o.g. that the number of faulty nodes is even.
            Split $F$ into four groups with $f/4$ processes each: $A$ and $B$, $C$ and $C'$.
            In addition, split $P$ into two groups $D$ and $D'$ with $f/2$ and $f/2+1$ processes respectively.
            In round $0$, every process FIFO broadcasts its initial value to everyone.

            [Round 1] At the beginning of the round, a single process $f_1 \in C$ receives and relay every value $v_1,\ldots,v_{f/4}$ (in this order) from processes in $A$.
            All processes in $A$ then crash.
            Similarly, a single process $f_1' \in C'$ receives and relay every value $v_1',\ldots,v_{f/4}'$ (in this order) from processes in $B$, which then crash.

            Subsequently, processes in $D$ and all remaining non-crashed processes in $C$ receive $v_1$ from $f_1$.
            Moreover, processes in $D'$ and all remaining non-crashed processes in $C'$ receive $v_1'$ from $f_1'$.
            Note that no process in $C \cup D$ received $v_1'$ and no process in $C' \cup D'$ received $v_1$.
            Both $f_1$ and $f_1'$ then crash.

            At the end of the round, every initial value from non-crashed processes (sent in round $0$) is received and relayed by non-crashed processes.

            [Round $i$ ($i \geq 2$)] At the beginning of the round, single process $f_i \in C$ receives and relays $v_i,\ldots,v_{f/4}$ from $f_{i-1}$ (resp. $f_i' \in C'$ receives and relays $v_i',\ldots,v_{f/4'}$ from $f_{i-1}$).

            Subsequently, processes in $D$ and all remaining non-crashed processes in $C$ receive $v_i$ from $f_i$ (but not $v_i'$).
            Processes in $D'$ and all remaining non-crashed processes in $C'$ receive $v_i'$ from $f_i'$ (but not $v_i$).
            Both $f_i$ and $f_i'$ then crash.
            Finally, every remaining value sent in the previous round by non-crashed processes are received (and relayed if applicable).

            \myparagraph{Output conditions} We use $C \cup D$ (resp. $C' \cup D'$) to refer to processes in $C \cup D$.
            By construction $|C \cup D| \leq |C' \cup D'| < f+1$.
            When $C \cup D$ receives $v_1$ in round $1$, it gives a clock value of $2$ to $v_1$ (similarly with $C' \cup D'$ and $v_1'$).
            In the end of the round, $C \cup D$ receives other initial values from non-crashed processes, but not $v_1'$, so $v_1'$ is attributed a higher clock value later.
            This ensures that $v_1'$ cannot be in the output without $v_1$, since $C \cup D$ never receives a quorum of forward messages for $v_1$ with clock value smaller than that of $v_1'$.
            But all the forward messages received later for $v_1'$ from $C' \cup D'$ have their clock values smaller for $v_1'$ than for $v_1$, so $v_1$ also cannot be in the output without $v_1'$.

            In addition, $C \cup D$ receives $v_i$ before receiving $v_{i-1}'$, attributing a smaller clock value to $v_i$.
            By the end of round $i$, $C \cup D$ has no quorum for which clock values are smaller for $v_{i-1}'$ than for $v_i$, and thus cannot output $v_{i-1}'$.
            This creates a chain of dependencies where $v_{i-1}$ cannot be in the output without $v_{i-1}'$ (and vice-versa), $v_{i-1}'$ cannot be in the output without $v_{i}$,
            and $v_i$ cannot be in the output because not enough forward messages for $v_i$ are received in round $i$.
            Therefore, in the end of round $i$, $C \cup D$ (and $C' \cup D'$) is unable to output a value.

            The execution described above can be extended for $f/4$ rounds.
        \end{proof}
    \end{theorem}

\section{Atomic Snapshot Operations}
\label{app:ASO}

    The papers~\cite{faleiro2012podc} and~\cite{imbs2018set} have a long-lived form of the algorithms in Appendix~\ref{app:oneShot}, for which one can use to implement $\AS$.
    In the following, we show that an $\ASO$ operation using~\cite{faleiro2012podc} has constant amortized time complexity, and thus conjecture that it has $O(k)$ time complexity in the worst-case.
    On the other hand,~\cite{imbs2018set}'s $\ASO$ operation latency is $O(n)$ even in fault-free runs.

    \subsection{Atomic Snapshot by Faleiro et al.~\cite{faleiro2012podc}}

    The Generalized Lattice Agreement ($\GLA$) described in Figure~\ref{fig:faleiro-la} splits the roles of the processes into proposers, acceptors and learners.
    For our purpose, we assume that every process performs the three roles.
    In addition, we add Algorithm~\ref{alg:bridgeLA} on top of the $\GLA$ protocol to match the interface used in Algorithm~\ref{alg:GLAtoAS}.

    \begin{figure}[!htp]
    \centering
    \includegraphics[width=\textwidth, page=6, trim = {5em 20em 13em 5em}, clip]{related_work/generalized-podc2012.pdf}
    \caption{LA algorithm as presented in~\cite{faleiro2012podc}.}
    \label{fig:faleiro-la}
\end{figure}

\begin{algorithm}
\begin{smartalgorithmic}[1]

\DistributedObjects
    \State $\GLA$ instance (Figure~\ref{fig:faleiro-la})
\EndDistributedObjects

\Operation{\Propose}{v}
    \State ReceiveValue(v)
    \State wait until $v \sqsubseteq $ LearntValue()
    \State \Return LearntValue()
\EndOperation
\end{smartalgorithmic}
\caption{Bridge protocol for Generalized Lattice Agreement~\cite{faleiro2012podc}.}
\label{alg:bridgeLA}
\end{algorithm}

    \begin{theorem}
        Consider the $\ASO$ protocol built from the composition of Algorithms~\ref{alg:GLAtoAS} and~\ref{alg:bridgeLA}.
        An operation takes at most $16$ rounds to complete if, during its interval, no correct process receives a message from a faulty one.

        \begin{proof}
            A message sent by a correct process is received by every correct process, and if a message sent in round $r$ is received, it must be received in at most round $r+1$ (from the definition of the metric).
            Since no message from faulty processes is received in the interval of the operation, we consider only events performed by correct ones.

            First, we show that once a process sends a proposal (line $24$ in Figure~\ref{fig:faleiro-la}) for a value $v$, all learners learn a value containing $v$ in at most $8$ rounds.

            Let $e_P$ be the event where process $i$ first sends a proposal for $v$, and let $0$ be the round assigned to it.
            By the end of round $1$, every (correct) process will have received $v$ and joined it in \textit{acceptedValue}, so that every \textbf{NACK} reply will now include $v$.
            As a consequence, every value learned from a proposal (or refinement) made after round $1$ must contain $v$.

            Suppose that some process already learned a value containing $v$ by the end of round $2$, then it received \textbf{ACK}s for this value from a majority of processes (which are correct).
            Every learner (thus, every process) receives the same \textbf{ACK}s within one round at most and is able to learn the same value.

            Now, if no process has already learned a value containing $v$, consider the InternalReceive($v$) message which is sent before the proposal.
            By the end of round $1$, every process has received the message and added $v$ to its buffer, and since no process had learned $v$ by the end of round $2$, every process must be proposing (i.e. \textit{status} = \textit{active}).

            Let $V$ be set of all active proposals in the end of round $2$, then by the end of round $3$ every acceptor will have received every value in $V$ and added it to \textit{acceptedValue}.
            So every reply made in round $4$ onward will contain all current values.
            If a process refines its proposal in round $5$, then it must have received at least one reply containing all values for the previous proposal, so by at most round $6$ all acceptors would reply with $ACK$ and all learners would learn a value by at most round $7$.

            Now consider the case where a process $j$ refines its proposal in round $4$, it may happen that the refined proposal still misses a value, in which case $j$ refines again in round $6$ (the latest) and this next proposal is guaranteed to include all values.
            Thus, all processes acknowledge the proposal by at most round $7$ and all learners are able to learn a value by round $8$.

            Let $e_C$ be the application call event received at a process $i$, $e_R$ its return event, and $v$ the value received for the operation.
            If $i$ is already active, it first buffers $v$ and waits until the current active proposal finishes before sending a proposal for $v$.
            Consider the worst case where $e_C$ happens just after $i$ started a new active proposal.
            As previously shown, it takes at most $8$ rounds until $i$ can propose a new value from \textit{bufferedValues} again, and once it proposes $v$ it can take another $8$ rounds at most to learn a value with it.
            In total, from the call event to the return event, there can be at most $16$ rounds.
        \end{proof}
    \end{theorem}

    \begin{corollary}
    Algorithms~\ref{alg:GLAtoAS} and~\ref{alg:bridgeLA} together have an amortized time complexity of $16$ rounds.
\end{corollary}
    
    \subsection{Atomic Snapshot by Imbs et al.~\cite{imbs2018set}}

    Imbs et al.~\cite{imbs2018set} use operations of the SCD-Broadcast protocol to implement atomic snapshot.
    As such, in the proof for Theorem~\ref{th:goodImbs} we build an execution that takes $\Omega(n)$ rounds for a process to output a value in the SCD-Broadcast protocol, implying that the same time complexity for a snapshot operation.

    Figure~\ref{fig:imbs-as} (extracted from~\cite{imbs2018set}) shows the algorithm for $\MW\MR$ $\ASO$ using SCD-Broadcast.
    The main difference to the $\SW\MR$ implementation is the addition of line $3$, which includes a ``read'' phase before updating the array and thus requires two SCD-Broadcast operations instead of one.
    As we only consider $\SW\MR$ $\ASO$ implementations, we assume that the only operations in the executions are snapshots (which is unchanged and requires a single SCD-Broadcast operation).

    \begin{figure}[!htp]
        \centering
        \includegraphics[page=7, trim = {11em 53em 11em 8em}, clip]{related_work/SCD.pdf}
        \caption{AS algorithm as presented in~\cite{imbs2018set}.}
        \label{fig:imbs-as}
    \end{figure}

    \begin{theorem}
    \label{th:goodImbs}
        A snapshot operation in~\cite{imbs2018set}'s protocol can take $\Omega(n)$ rounds in fault-free runs.

        \begin{proof}
            First consider an execution of SCD-Broadcast with an \emph{even} number of processes.
            We proceed to build an execution where an operation takes $n$ rounds to complete.
            We split the system into two groups $A$ and $B$ with $n/2$ processes each.
            Note that neither $A$ nor $B$ alone form a quorum.
            We also say $A$ or $B$ to refer to all processes in $A$ or $B$.
            In the execution below, every time a process in $A$ (resp. B) replies a value (sends a forward message for it), all the processes in $A$ receive it immediately after (similar for B).

            [Round $0$] A single process $a_0 \in A$ sends a forward message with $v_0^A$ to everyone.

            [Round $1$] At the beginning of the round, $A$ receives $v_0^A$ and forwards the value.
            Subsequently, a process $b_0 \in B$ sends a new forward message for value $v_0^B$, which is received and relayed right away by $B$ (before $v_0^A$).
            At the end of the round, $B$ then receives $v_0^A$ from $A$ and relays it, but although there is a quorum for $v_0^A$, no quorum has each clock assignment for $v_0^A$ smaller then that of $v_0^B$, thus $B$ cannot output $v_0^A$ without $v_0^B$.
            Since there is no quorum of replies for $v_0^B$, $B$ cannot output.

            [Round $2\cdot i$] At the beginning of the round, a new process $a_i \in A$ sends forward with $v_i^A$, received right away (before $v_{i-1}^B$ from $B$) by $A$, which relays it.
            Subsequently, $A$ receives $v_{i-1}^B$ and relays it.
            $A$ is unable to output $v_{i-1}^A$ without $v_{i-1}^B$ since $B$ assigned a smaller clock value to $v_{i-1}^B$, and is unable to output $v_{i-1}^B$ without $v_i^A$ since it assigned a smaller clock value to $v_i^A$, and there is no quorum of replies received for $v_i^A$.

            [Round $2\cdot i + 1$] At the beginning, a process $b_i \in B$ sends forward with $v_i^B$, received right away by $B$ (before $v_i^A$) which relays it.
            Subsequently, $B$ receives $v_i^A$ and the reply for $v_{i-1}^B$ from $A$ in this order.
            But $B$ cannot output $v_{i-1}^B$ without $v_i^A$, since there is no quorum assigning a smaller clock value to $v_{i-1}^B$ than to $v_i^A$.
            But $v_i^A$ cannot be in the output without $v_i^{B}$ either, for which $B$ does not have a quorum of replies.
            $B$ is therefore unable to output.

            Using the steps above we can delay the execution up to $n$ rounds.
            When the number of processes is \emph{odd}, we split the system into $3$ groups: $A$, $B$ and $C$, and proceed in a similar fashion as above for $A$ and $B$, but a new process from $C$ now has initiate a new value in the beginning of every turn in order to delay the execution.
            This construction can delay the execution up to $n/3$ rounds. The execution proceeds as following:

            [Round $0$] A single process $a_0 \in A$ sends a forward message with $v_0^A$ to everyone.
            A single process $c_0 \in C$ sends a forward message with $v_0^C$ to everyone.

            [Round $1$] At the beginning of the round, $A$ receives $v_0^A$ and forwards the value, $C$ receives $v_0^C$ and relays it.
            Subsequently, a process $b_0 \in B$ sends a new forward message for value $v_0^B$, which is received and relayed right away by $B$ (before $v_0^A$ or $v_0^C$).
            Also, another process $c_1 \in C$ sends a forward message with $v_1^A$ to everyone, which $C$ receives and forward immediately after.
            At the end of the round, $B$ receives $v_0^A$ from $A$ and $v_0^C$ from $C$ and relays them, but although there is a quorum for $v_0^A$ and $v_0^C$, no quorum has each clock assignment for $v_0^A$ or $v_0^C$ smaller then that of $v_0^B$, thus $B$ cannot output $v_0^A$ and $v_0^C$ without $v_0^B$.
            Since there is no quorum of replies for $v_0^B$, $B$ cannot output.
            Similarly, $A$ receives $v_0^C$ and $C$ receives $v_0^A$ but they cannot output.

            [Round $2\cdot i$] At the beginning of the round, a new process $a_i \in A$ sends forward with $v_i^A$, received right away (before $v_{i-1}^B$ from $B$ and $v_{2\cdot i - 1}^C$ from $C$) by $A$, which relays it.
            Similarly, a process $c_{2\cdot i} \in C$ sends forward with $v_{2\cdot i}^C$ before $C$ receives $v_{i-1}^B$ from $B$.

            Subsequently, $A$ receives $v_{i-1}^B$ and $v_{2\cdot i - 1}^C$ and relays them.
            $A$ is unable to output $v_{i-1}^A$ without $v_{i-1}^B$ or $v_{2\cdot i-1}^C$ since $B$ assigned a smaller clock value to $v_{i-1}^B$ and $C$ assigned a smaller clock value to $v_{2\cdot i-1}^C$, and is unable to output $v_{i-1}^B$ and $v_{2\cdot i - 1}^C$ without $v_i^A$ since it assigned a smaller clock value to $v_i^A$, and there is no quorum of replies received for $v_i^A$.
            Moreover, $C$ receives $v_{i-1}^B$ from $B$.
            $C$ cannot output $v_{2\cdot i - 2}^C$ without either $v_{i-1}^A$ or $v_{i-1}^B$ because $A$ assigned a smaller value to $v_{i-1}^A$ and $B$ assigned a smaller value to $v_{i-1}^B$.
            But both cannot be output without $v_{2\cdot i -1}^C$, which was assigned a smaller value, and there is no quorum for $v_{2\cdot i -1}^C$ at $C$.

            [Round $2\cdot i + 1$] At the beginning, a process $b_i \in B$ sends a forward message with $v_i^B$, received right away by $B$ (before $v_i^A$ or $v_{2\cdot i}^C$) which relays it.
            A process $c_{2\cdot i + 1}$ sends forward with $v_{2\cdot i + 1}$, before receiving $v_iA$.

            Subsequently, $B$ receives $v_i^A$ and $v_{2\cdot i}^C$ as well as the replies for $v_{i-1}^B$ from $A$ and $C$ in this order.
            But $B$ cannot output $v_{i-1}^B$ without $v_i^A$ or $v_{2\cdot i}^C$, since $A$ and $C$ assigned smaller clock values to $v_i^A$ and $v_{2\cdot i}^C$ respectively.
            But neither $v_i^A$ nor $v_{2\cdot i}^C$ can be in the output without $v_i^{B}$, for which $B$ does not have a quorum of replies.
            Now, $C$ receives $v_i^A$ and both replies for $v_{2\cdot i-1}$ from $B$ and $A$ in this order.
            But $A$ assigned a smaller value to $v_i^A$ and $B$ also to $v_{i-1}^B$,
            and neither can be output without $v_{\cdot i}$, for which $C$ has no quorum.
        \end{proof}
    \end{theorem}

\begin{corollary}
    A snapshot operation in~\cite{imbs2018set}'s protocol can take $\Omega(n)$ rounds in the worst-case.
\end{corollary}

\end{document}